\pdfoutput=1
\documentclass[11pt]{article}

\usepackage[margin=1in]{geometry}
\usepackage{amsmath,amssymb,amsthm,mathtools}
\usepackage{microtype}
\usepackage{aliascnt}
\usepackage{booktabs}
\usepackage{enumitem}
\usepackage{algorithm}
\usepackage[noend]{algpseudocode}
\usepackage{url}
\usepackage[hidelinks]{hyperref}
\usepackage[nameinlink,noabbrev]{cleveref}
\usepackage{tikz}

\newtheorem{theorem}{Theorem}[section]
\newaliascnt{lemma}{theorem}
\newtheorem{lemma}[lemma]{Lemma}
\aliascntresetthe{lemma}
\newaliascnt{proposition}{theorem}
\newtheorem{proposition}[proposition]{Proposition}
\aliascntresetthe{proposition}
\newaliascnt{claim}{theorem}
\newtheorem{claim}[claim]{Claim}
\aliascntresetthe{claim}
\newaliascnt{corollary}{theorem}

\aliascntresetthe{corollary}
\theoremstyle{definition}
\newaliascnt{definition}{theorem}
\newtheorem{definition}[definition]{Definition}
\aliascntresetthe{definition}
\theoremstyle{remark}
\newaliascnt{remark}{theorem}

\aliascntresetthe{remark}

\newcommand{\eps}{\varepsilon}
\newcommand{\OPT}{\mathrm{OPT}}
\newcommand{\obj}{\operatorname{obj}}
\newcommand{\cost}{\operatorname{cost}}
\newcommand{\vol}{\operatorname{vol}}
\newcommand{\TV}{\mathrm{TV}}

\newcommand{\I}{\mathrm{I}}
\newcommand{\Hent}{\mathrm{H}}
\newcommand{\one}{\mathbf{1}}
\newcommand{\E}{\mathbb{E}}
\newcommand{\Prb}{\mathbb{P}}
\newcommand{\Kcal}{\mathcal{K}}
\newcommand{\Ccal}{\mathcal{C}}
\newcommand{\Acal}{\mathcal{A}}

\newcommand{\poly}{\mathrm{poly}}

\title{Pseudometric-Weighted Correlation Clustering  via Spectral Preclustering}
\author{
Chenglin Fan\thanks{Seoul National University.}
\and
Dahoon Lee\thanks{New York University.}
\and
Euiwoong Lee\thanks{University of Michigan.}
}
\date{}

\Crefname{claim}{claim}{claims}
\Crefname{claim}{Claim}{Claims}
\Crefname{proposition}{proposition}{propositions}
\Crefname{proposition}{Proposition}{Propositions}
\Crefname{lemma}{lemma}{lemmas}
\Crefname{lemma}{Lemma}{Lemmas}
\Crefname{theorem}{theorem}{theorems}
\Crefname{theorem}{Theorem}{Theorems}
\makeatletter
\providecommand{\theHALG@line}{}
\renewcommand{\theHALG@line}{\arabic{algorithm}.\arabic{ALG@line}}
\makeatother

\begin{document}
\maketitle

\begin{abstract}
We study \emph{pseudometric-weighted correlation clustering}, in which
pairs have arbitrary positive or negative labels and nonnegative
disagreement weights satisfying the triangle inequality. For every
fixed $\eps>0$, we give a randomized polynomial-time
$(2+\eps)$-approximation, improving the previous factor of $10/3$.
The algorithm constructs an explicit feasible solution to the cluster LP
with polynomial support and value at most $(1+\eps)\OPT$, where
$\OPT$ is the minimum integral clustering cost.

Our main contribution is a weighted preclustering construction that witnesses a near-optimal clustering while bounding the total weight of admissible edges by $\operatorname{poly}(1/\eps)\OPT$.
We obtain admissibility certificates from random walks, where the factor-two Ptolemy's inequality for pseudometrics reduces the required length of walks from $O(\log n)$ to a constant.
We then bound the rounding error separately at each vertex and use the triangle inequality to control its weighted cost.
Sampling and coverage corrections produce an exactly feasible solution, and the factor-two rounding completes the approximation algorithm.
\end{abstract}

\section{Introduction}

Correlation clustering seeks a partition consistent with pairwise
similarity and dissimilarity information, without prescribing the
number of clusters. Let $V$ be a set of $n$ vertices, and let
$E^+$ and $E^-$ partition the undirected edges in $E$ into positive and
negative signs. A clustering pays for a positive pair if it separates
its endpoints, and for a negative pair if it places them together.
Since the formulation of Bansal, Blum, and Chawla~\cite{bansal2002correlation},
the complete-graph problem has motivated approximation algorithms based
on pivoting, LP rounding, and stronger relaxations
\cite{charikar2005clustering,ailon2008aggregating,chawla2015near,
cohenaddad2022sherali,cohenaddad2023preclustering}.

We study \emph{pseudometric-weighted correlation clustering}.
Every unordered pair of distinct vertices now has a disagreement penalty as well as a sign.
We impose the triangle inequality on the penalties; the signs remain arbitrary.
Denoting $x^{\Ccal}_{uv}$ as an indicator that a partition $\Ccal$ separates
$u$ and $v$, its disagreement cost is
\[
  \obj(\Ccal)=\sum_{uv\in E^+}w_{uv}x^{\Ccal}_{uv}
             +\sum_{uv\in E^-}w_{uv}(1-x^{\Ccal}_{uv}),
  \qquad \OPT:=\min_{\Ccal}\obj(\Ccal),
\]
where $\OPT$ is the optimal {integral} cost.
The original correlation clustering problem on the complete graph is the unweighted special case $w_{uv}=1$ for all $u\ne v$, while that on the complete $k$-partite graph~\cite{chawla2015near} is the special case $w_{uv}=0$ or $1$ for all $u\ne v$.
For unrestricted weights, the problem admits an
$O(\log n)$-approximation~\cite{demaine2006weighted} and is UGC-Hard to approximate within any constant factor~\cite{chawla2006multicut}.
Under the triangle inequality, Charikar and Gao~\cite{charikar2024umvd} obtained the first constant-factor approximation.
Fan, Lee, and Lee~\cite{fan2025pseudometric}
improved the factor to $10/3$ and established a matching limitation for their standard-LP pivoting analysis. We improve the factor to
$2+\eps$ by using the stronger \emph{cluster LP}.

\subsection{Our results}

\begin{theorem}[Main approximation theorem]\label{thm:main}
There exists a randomized algorithm for pseudometric-weighted
correlation clustering that runs in time $n^{\poly(1/\eps)}$ and
returns a clustering of cost at most $(2+\eps)\OPT$ with high
probability.
\end{theorem}

We extend the unweighted
cluster-LP construction of Cao, Cohen-Addad, Lee, Li, Newman, and
Vogl~\cite{cao2024clusterlp} to pseudometric weights.

\begin{theorem}[Solving cluster LP with pseudometric weights]
\label{thm:clusterlp-main}
Let $\eps>0$ be any sufficiently small constant.
There exists an $n^{\poly(1/\eps)}$-time randomized algorithm that outputs a feasible cluster LP solution to the given pseudometric-weighted correlation clustering
instance with an objective value at most $(1+\eps)\OPT$ with high probability.
\end{theorem}

Since an explicit cluster LP
solution can be rounded at factor $2$ for arbitrary nonnegative pair
weights (see \Cref{lem:rounding}), \Cref{thm:main} directly follows from \Cref{thm:clusterlp-main}. The main difficulty is thus to
{construct} a near-optimal feasible fractional solution, with its
error measured relative to $\OPT$ rather than to the input weight.

Although it is immediate to obtain a truly sub-$2$ approximation by combining the result with the existing constant-factor pivot-based approximation algorithm as the original unweighted problem does,
we omit the improvement in this paper as our main focus is to extend the framework to general settings.

\subsection{The cluster-LP approach and the weighted obstacles}

\paragraph{Cluster LP formulation.}
The cluster LP has a nonnegative variable $z_S$ for every nonempty
$S\subseteq V$, interpreted as fractional mass assigned to $S$ as a
standalone cluster. It requires $\sum_{S\ni u}z_S=1$ for each vertex and interprets the separating indicators $x_{uv}^\Ccal$ as
$x_{uv}=1-\sum_{S\supseteq\{u,v\}}z_S$; its objective is the same
weighted disagreement expression as above. Thus, the LP formulation
already permits weights, which only changes the objective function linearly. The result of~\cite{cao2024clusterlp} for the
unweighted case is an efficient construction of a feasible solution
of value at most $(1+\eps)\OPT$, without directly optimizing over
exponentially many cluster variables.

\paragraph{The original construction.}
Preclustering begins by computing a partition $\Kcal$ of $V$ into
\emph{atoms}. An atom $K\in\Kcal$ is a set of vertices treated as
\emph{indivisible}: a clustering $\Ccal$ \emph{respects} $\Kcal$ if
every $K\in\Kcal$ is wholly contained in one part of $\Ccal$.
Consequently, a cluster is a union of atoms, possibly several of them.
An atom need not itself be an optimal cluster; its vertices need not
have zero distances or exclusively positive labels. In both unweighted and our weighted constructions, atoms are obtained from an initial constant-factor
clustering: vertices with large incident disagreement cost are isolated,
whereas sufficiently consistent vertices can remain together as a
\emph{core atom}; an initial cluster is split completely into singletons if the isolating
vertices carry too much disagreement. The resulting partition has
cost $O(\OPT)$, while ensuring {some} optimal clustering respects it.

Preclustering also computes a set $E_{\mathrm{adm}}$ of
\emph{admissible edges} joining distinct atoms. The
decision is made for an entire atom pair: for distinct atoms $A,B$,
either all pairs in $A\times B$ are admissible or none are.
We show that some near-optimal clustering $\Ccal_1^\star$ respects the atoms and distinct atoms are located in the same cluster only if they are connected with admissible edges: We say that the preclustering \emph{witnesses} $\Ccal_1^\star$.
For vertex
sets $A,B$, write $w(A,B)=w(A\times B)=\sum_{a\in A,b\in B}w_{ab}$, and use
$w^+(A,B)$ and $w^-(A,B)$ for its positive and negative parts: $w((A\times B)\cap E^+)$ and $w((A\times B)\cap E^-)$.
Since it is defined as an ordered set-pair sum, $w_{ab}$ is counted twice whenever $a,b\in A\cap B$.

With these notions, the unweighted construction follows three stages.
\begin{enumerate}
    \item \textbf{Preclustering} preserves a near-optimal atom-respecting witness $\Ccal_1^\star$
    and bounds the \emph{number} of admissible edges by
    $\operatorname{poly}(1/\eps)\OPT$
    \cite{cohenaddad2023preclustering,cao2024clusterlp}.
    \item Solve a \textbf{bounded sub-cluster LP}, a polynomial-size relaxation that records cluster-size-indexed moments up to order $r = \poly(1/\eps)$, with constraints enforcing the preclustering.
    \item \textbf{Correlated rounding}~\cite{raghavendra2012cardinality}
    generates actual sets from these local moments, by growing the seed using up to $r-1$ conditional steps and then applying the independent rounding conditioned on the seed; repeated sampling and
    coverage correction produce the explicit cluster-LP solution.
\end{enumerate}


\paragraph{Two obstacles in the weighted setting.}
Both obstacles occur at the point where the original construction utilizes the uniformity of weights.
The first is a \emph{combinatorial} obstruction in constructing admissible edges: the computable admissibility
test must witness some near-optimal clustering while ensuring $w(E_{\mathrm{adm}})\le\operatorname{poly}(1/\eps)\OPT$.
Declaring all pairs admissible loses the weight bound, whereas simply discarding expensive pairs may destroy every near-optimal witness.
The unweighted construction for admissible edges uses the fact that the quadratic form $\sum_{p\in C}w^+(u,p)w^+(v,p)$ within the same cluster $C\in\Ccal_1^\star$ is bounded below in terms of $w^+(u,V)$ and $w^+(v,V)$.

However, in the weighted case, this bound no longer holds, as shown in \Cref{fig:bipartite}:
$V$ consists of two parts $A$ and $B$, each of size $n$;
Every edge between $A$ and $B$ has weight $1$, while every edge inside $A$ or $B$ has weight $\Delta\leq 1$;
$E$ forms a $\beta n$-regular bipartite expander graph between $A$ and $B$, where $\beta>0$ is a sufficiently small constant.
In this setting, merging all vertices is near-optimal, while separating all is at most $O(\frac{1-\beta}{\beta})$ times worse, hence valid atoms.
However, the quadratic form for $u\in A$ and $v\in B$ is at most $2n\Delta$ since one of $w^+(u,p)$ and $w^+(v,p)$ is at most $\Delta$ for all $p\in A\cup B$, while both $w^+(u,A\cup B)$ and $w^+(v,A\cup B)$ are at least $(1-\beta)n$.
By shrinking the diameter $\Delta$ to $0$, this implies that the quadratic form can be arbitrarily small compared to two terms.
Therefore, we cannot trivially use the same construction and analysis as in the unweighted case.
\begin{figure}
    \centering
    
    \begin{tikzpicture}[
        vertex/.style={circle,fill=black,inner sep=0pt,minimum size=5pt},
        line width=0.8pt
    ]
        \draw[dashed] (0,0) circle (0.9);
        \draw[dashed] (4.5,0) circle (0.9);
        \node at (0,1.1) {$A$};
        \node at (4.5,1.1) {$B$};
        \node[vertex,label=left:$u$] (u) at (0,0) {};
        \node[vertex,label=left:$p$] (p) at (0,-0.6) {};
        \node[vertex,label=right:$v$] (v) at (4.5,0) {};
        \draw (u) -- (p) -- (v) -- (u);
        \draw[<->,>=stealth] (0.4,1.1) -- node[above] {$1$} (4.1,1.1);
        \draw[<->,>=stealth] (-0.8,-1.1) -- node[below] {$\Delta$} (0.8,-1.1);
        \draw[<->,>=stealth] (3.7,-1.1) -- node[below] {$\Delta$} (5.3,-1.1);
    \end{tikzpicture}
    \caption{Correlation clustering instance with almost-bipartite metric weights.}
    \label{fig:bipartite}
\end{figure}

The second is a \emph{probabilistic} obstruction in correlated rounding:
the same rounding scheme should also satisfy the concentration bound with respect to the weighted objective.
The unweighted construction combines a bound of order
$|N(a)|^2/\sqrt r$ on the aggregated error with an auxiliary
cluster-size restriction of the form
$s>|K_a|+\eta|N(a)|$, where $a$ is the first sampled vertex, $K_a$ is the atom containing $a$, $N(a)$ is a neighborhood set of $a$ with respect to admissible edges, and $\eta>0$ is an accuracy parameter;
if the sampled set is just $K_a$, the pairwise error is zero.
The size restriction allows the factor $1/s$ in the first-seed
probability $\Prb[s,u_1=a]=y_a^s/(s y_\emptyset)$ to cancel one
factor of $|N(a)|^2$. Using
$\sum_s y_a^s=1$ then controls the unweighted objective by a term
of order $|E_{\mathrm{adm}}|/(\eta y_\emptyset\sqrt r)$
\cite{cao2024clusterlp}. With weighted penalties, neither the
cardinality bound on $E_{\mathrm{adm}}$ nor the average marginal-error
bound controls the {weighted} objective.
Therefore, a marginal-error estimate for
{each fixed endpoint} is needed, not merely an unweighted average over
$N(a)\times N(a)$.

\subsection{Technical overview}
We follow the framework of \cite{cao2024clusterlp} along with new techniques to resolve obstacles and solve the cluster LP.
Our preclustering test compares constant-length random walks on the atom graph---we call this \emph{spectral preclustering} because its admissibility test uses short random-walk overlaps, whose guarantees follow from spectral properties of the witness clusters;
the sampling analysis controls error separately at every vertex.
\paragraph{Near-optimal clustering from a random walk perspective.}
We consider a random walk using $w^+$ on a near-optimal clustering $\Ccal_1^\star$.
Starting with the atom partition $\Kcal$, we choose for the analysis
an atom-respecting optimum and refine it by splitting unions of atoms,
with total additional cost at most $\eps\OPT$. The resulting witness
$\Ccal_1^\star$ has two spectral properties in each cluster $C$ containing at
least two atoms. Every nontrivial cut into unions of atoms has more
positive than negative weight, and each atom retains a
$\Theta(\eps)$ fraction of its global
degree. The former condition, together with the original triangle
inequality, gives constant lower bound of conductance for the walk restricted to $C$; the latter permits comparison with a global walk.

\paragraph{From Ptolemy's inequality to quadratic form witness and short-walk certificate.}
As quadratic forms play a central role in this perspective, we require useful bounds for the product of two $w^+$ terms.
The Ptolemy's inequality, an inequality that involves products of two Euclidean distances, is therefore considered.
Although it is not true in general pseudometric, it remains true up to a constant factor (\Cref{prop:ptolemy}).
Using this property, we provide a \emph{blueprint}~(\Cref{clm:pair-volume}) of our admissibility test analogous to the original construction.

To make the test computable, it is necessary to obtain a proxy for the components of the blueprint, since $C\in\Ccal_1^\star$ on the right-hand side is still not accessible.
As the stationary distribution for a random walk restricted to $C$ contains those values and such walk is comparable with a global walk,
we can use a global walk to obtain it, where the constant lower bound of conductance guarantees exponential $L_2$ convergence.
However, since the initial $L_2$ distance can reach $O(\sqrt n)$,
the convergence bound alone gives only an $O(\log n)$-length walk certificate.
To resolve this, we show that the distance after a first walk is at most constant using Ptolemy's inequality (\Cref{clm:warm-start}),
thus a constant-length walk suffices.

\paragraph{Lazy random walk resolves almost-bipartite cases.}
The bad examples where the original construction fails, as in \Cref{fig:bipartite}, are due to the possibility that some eigenvalue can be arbitrarily close to $-1$, while the construction only uses a regular random walk.
Therefore, we apply a lazy random walk $P_{\mathrm{lazy}}=\frac{I+P}{2}$
to resolve this issue.

\paragraph{Correlated rounding and exact feasibility.}
We solve a bounded sub-cluster LP with the atom and admissibility
constraints, retaining the size-indexed moments but not the original
auxiliary neighborhood-cardinality restriction.
While the sampling process remains almost the same, we refine the marginal-error estimate analysis to make it applicable to the pseudometric-weighted setting;
the only change made from the process is rather to stop growing at a uniformly random horizon than whenever the mutual information is small enough, which only simplifies the analysis.

For any $s$ and $a$ given, where $s$ is a sampled size index and $a$ is the first sampled vertex, the original analysis bounds the total marginal-error estimate $\sum_{u,v}e_{uv}^{s,a}=O(|N(a)|^2/\sqrt{r})$.
We provide a stronger claim: For any $s$, $a$, \emph{and} $u$ given,
we bound the marginal-error estimate $\sum_{v}e_{uv}^{s,a}=O(s/\sqrt{r})$ (\Cref{clm:conditional-pair-error}).
Combining it with triangle-inequality charging through $a$,
and averaging using
$\Prb[s,u_1=a]=y_a^s/(s y_\emptyset)$, yields the total weighted error bound.

\paragraph{Relation to other cluster-LP results.}
The unweighted cluster-LP construction has been accelerated in
\cite{cao2025sublinear}. Garc\'{\i}a-Soriano and
Schohn~\cite{garciasoriano2026dual} give approximate dual separation
relative to the fractional optimum, improved unweighted rounding with
ratio $1.3865+\eps$, and a separation extension to weights in a fixed
positive range. Our guarantee instead uses pseudometricity to allow
zero weights and arbitrary positive weight ratios; its benchmark
remains the integral optimum. For chromatic correlation clustering,
Fan, Lee, and Lee~\cite{fan2026chromaticclusterlp} give a
$1.64$-approximation using chromatic cluster LP, while Abbasi, An, Byrka, Lee, and
Shin~\cite{abbasi2026weightedchromatic} give a $(2+\eps)$-approximation
also covering weighted instances under probability constraints.
Those constraints differ from the pseudometric penalties here.

\paragraph{Organization.}
\Cref{sec:prelim} collects metric inequalities and the cluster LP.
\Cref{sec:precluster} proves spectral preclustering.
\Cref{sec:boundedlp,sec:sampling,sec:construction} develop the
bounded sub-cluster LP, weighted sampler, and explicit LP construction.
The atom-preservation exchange argument appears in
\Cref{lem:atom-respecting}; \Cref{app:atoms} records its cost identities.

\section{Preliminaries}\label{sec:prelim}

\subsection{Pseudometric-weighted correlation clustering}

Let $V$ be a set of $n$ vertices. Every unordered pair of distinct
vertices is labeled positive or negative; thus
$\binom{V}{2}=E^+\uplus E^-$. Each pair $uv$ has a nonnegative
weight $w_{uv}=w_{vu}$, and $w_{uu}=0$. We assume that $w$ is a
\emph{pseudometric}, meaning that
\begin{equation}\label{eq:triangle}
  w_{uv}\le w_{ua}+w_{av}
  \qquad\text{for all }u,v,a\in V.
\end{equation}
Signs are arbitrary, zero distances are allowed, and positive weight
ratios are unrestricted. Running times include the encoding length of
rational weights. The unweighted case has $w_{uv}=1$ for $u\ne v$.

A \emph{clustering} is a partition $\Ccal$ of $V$. Write
$C_u$ for the part containing $u$, and let
$x^{\Ccal}_{uv}=1$ if $C_u\ne C_v$ and $0$ otherwise.
The disagreement cost and its optimal integral value are
\begin{equation}\label{eq:mwcc-obj}
  \obj(\Ccal)
  :=\sum_{uv\in E^+}w_{uv}x^{\Ccal}_{uv}
    +\sum_{uv\in E^-}w_{uv}(1-x^{\Ccal}_{uv}),
  \qquad
  \OPT:=\min_{\Ccal}\obj(\Ccal).
\end{equation}

For sets $A,B\subseteq V$, define
\[
  w(A,B):=\sum_{a\in A,b\in B}w_{ab},\quad w^+(A,B):=\sum_{\substack{a\in A,b\in B\\ab\in E^+}}w_{ab},\quad w^-(A,B):=\sum_{\substack{a\in A,b\in B\\ab\in E^-}}w_{ab}.
\]
These sums allow overlapping or empty sets. They are ordered:
$w(A,A)$ counts internal pairs twice, whereas sums over $\binom V2$
count them once. We abbreviate $w(\{u\},A)$ to $w(u,A)$, and likewise
for the signed sums.

Let $N_u^+:=\{v\ne u:uv\in E^+\}$.
For any $u\in V$ and $S\subseteq V$, write
\[
  \cost_u(S):=w^+(u,V\setminus S)+w^-(u,S).
\]
For clustering $\Ccal$, its \emph{incident disagreement weight} at $u$ is
\begin{equation}\label{eq:incident-disagreement}
  w(u,N_u^+\triangle C_u)
  =w^+(u,V\setminus C_u)+w^-(u,C_u).
\end{equation}
(The symmetric difference includes $u$, but $w_{uu}=0$.) Each
disagreed pair is incident to exactly two vertices, so
\begin{equation}\label{eq:vertex-cost-identity}
  \obj(\Ccal)
  =\frac12\sum_{u\in V}
     w\bigl(u,N_u^+\triangle C_u\bigr).
\end{equation}
All accuracy parameters used in the construction belong to
$(0,1]$; larger target accuracies can be replaced by $1$.

\subsection{Metric inequalities for sets}

The following inequalities allow zero distances and overlapping sets.

\begin{proposition}[Set-wise triangle inequality]\label{prop:settriangle}
For nonempty $A,B,C\subseteq V$,
\[
  \frac{w(A,B)}{|A||B|}
  \le \frac{w(A,C)}{|A||C|}
      +\frac{w(B,C)}{|B||C|}.
\]
In particular,
\[
  \frac{w(A,A)}{|A|}\le 2\frac{w(A,C)}{|C|}
\]
and $w(A,A)\le 2|A|w(u,A)$ for any $u\in V$.
\end{proposition}
\begin{proof}
Sum $w_{ab}\le w_{ac}+w_{cb}$ over
$(a,b,c)\in A\times B\times C$ and divide by $|A||B||C|$.
Taking $B= A$ gives the second, and further taking $C=\{u\}$ gives the final assertion.
\end{proof}


Although general metrics need not satisfy the usual Ptolemy's inequality as we can observe from the unit four-vertex cycle graph~(see \Cref{fig:4cycle}), the following factor-two form follows
from the triangle inequality.
\begin{figure}
    \centering
    \begin{tikzpicture}[
        vertex/.style={circle, fill=black, inner sep=0pt, minimum size=5pt},
        line width=0.8pt
    ]
        \node[vertex, label=above:$a$] (A) at (0,1) {};
        \node[vertex, label=below:$b$] (B) at (0,-1) {};
        \node[vertex, label=left:$c$]  (C) at (-1,0) {};
        \node[vertex, label=right:$d$] (D) at (1,0) {};
        \draw (A) -- (C) -- (B) -- (D) -- (A);
    \end{tikzpicture}
    \caption{The shortest-path metric of a four-cycle with unit-length edges. This is also a tight instance for \Cref{prop:ptolemy}, since $w_{ab}w_{cd}=4$ and $w_{ac}w_{bd}+w_{ad}w_{bc}=2$.}
    \label{fig:4cycle}
\end{figure}
\begin{proposition}[Factor-two Ptolemy's inequality for pseudometrics]\label{prop:ptolemy}
For any $a,b,c,d\in V$,
\[
  w_{ab}w_{cd}
  \le 2\bigl(w_{ac}w_{bd}+w_{ad}w_{bc}\bigr).
\]
Consequently, for arbitrary sets $A,B,C,D\subseteq V$,
\begin{equation}\label{eq:set-ptolemy}
  w(A,B)w(C,D)
  \le 2\bigl(w(A,C)w(B,D)+w(A,D)w(B,C)\bigr),
\end{equation}
and, in particular,
\begin{equation}\label{eq:ptolemy-special}
  w(A,B)w(C,C)\le 4w(A,C)w(B,C).
\end{equation}
\end{proposition}
\begin{proof}
The expression on each side is invariant under interchanging
$a,b$, interchanging $c,d$, or exchanging these two pairs.
We may therefore assume that $x:=w_{ac}$ is the smallest of
the four cross distances
$w_{ac},w_{ad},w_{bc},w_{bd}$.
Put $y:=w_{bc}$ and $z:=w_{ad}$. The triangle inequality gives
\[
  w_{ab}w_{cd}\le (x+y)(x+z).
\]
Since $x\le y,z$, we have $(y-x)(z-x)\ge0$, so
\[
  (x+y)(x+z)\le2x^2+2yz
  \le2\bigl(w_{ac}w_{bd}+w_{ad}w_{bc}\bigr),
\]
where the last step uses $x\le w_{bd}$.
Summing the resulting pointwise inequality over
$A\times B\times C\times D$ proves \eqref{eq:set-ptolemy}.
Taking $D=C$ yields \eqref{eq:ptolemy-special}.
\end{proof}


\paragraph{Zero-diagonal block weights.}
If $K$ is nonempty and nonsingleton, then $w(K,K)$ and $w^+(K,K)$ can be positive. These weights act as self-loops in the auxiliary graph where each atom corresponds to a vertex.
In this auxiliary graph, we retain only interactions between distinct atoms, so we introduce $\omega$ that modifies $w$ by setting all $\omega(K,K)$ to zero.

Let $\mathcal P$ be a partition of $V$ into nonempty blocks. For
$A,B\in\mathcal P$, define
\begin{equation}\label{eq:block-kernels}
  \omega(A,B):=\begin{cases}w(A,B),&A\ne B,\\0,&A=B,\end{cases}
  \qquad
  \omega^+(A,B):=\begin{cases}w^+(A,B),&A\ne B,\\0,&A=B.\end{cases}
\end{equation}
For subfamilies $\mathcal S,\mathcal T\subseteq\mathcal P$, write
$\omega(\mathcal S,\mathcal T)=\sum_{A\in\mathcal S,B\in\mathcal T}
\omega(A,B)$ and $\omega(A,\mathcal T)=\sum_{B\in\mathcal T}\omega(A,B)$;
use the same convention for $\omega^+$. These are sums over blocks,
whereas $w(X,Y)$ always denotes the original sum over vertices.
In particular, $\omega(A,A)=0$ even when $w(A,A)>0$.

In fact, we can realize $\omega$ as the block weights of another pseudometric on the original vertex set $V$.
For each $a\in V$, let $P_a$ denote the block containing $a$, and define
\[w'_{ab}:=
\begin{cases}
\frac{w(P_a,P_b)}{|P_a||P_b|},&P_a\ne P_b;\\
0,&P_a=P_b.
\end{cases}\]
In other words, contract each $A\in\mathcal P$ to a point, and set the distance between points as the average original distance.
Therefore, applying \Cref{prop:ptolemy} to $w'$ yields the following proposition:
\begin{proposition}[Ptolemy inequality for the zero-diagonal block kernel]
\label{prop:block-ptolemy}
For all $A,B,X,Y\in\mathcal P$, allowing repeated blocks,
\begin{equation}\label{eq:block-ptolemy}
  \omega(A,B)\omega(X,Y)
  \le2\bigl(\omega(A,X)\omega(B,Y)
             +\omega(A,Y)\omega(B,X)\bigr).
\end{equation}
Consequently, for every $\mathcal U\subseteq\mathcal P$,
\begin{equation}\label{eq:block-ptolemy-sum}
  \omega(A,B)\omega(\mathcal U,\mathcal U)
  \le4\omega(A,\mathcal U)\omega(B,\mathcal U).
\end{equation}
\end{proposition}


\subsection{Cluster LP}

The cluster LP~\cite{cao2024clusterlp} assigns nonnegative mass $z_S$
to each nonempty set $S\subseteq V$ and defines the fractional separation
$x_{uv}$ by
\begin{align}
  \sum_{S\ni u}z_S&=1
    &&\forall u\in V,\label{eq:clp-vertex}\\
  \sum_{S\supseteq\{u,v\}}z_S&=1-x_{uv}
    &&\forall uv\in\binom{V}{2},\label{eq:clp-pair}\\
  z_S&\ge0
    &&\forall S\subseteq V.\label{eq:clp-nonneg}
\end{align}
Coverage and nonnegativity imply
$0\le x_{uv}\le1$. The objective is
\begin{equation}\label{eq:cluster-lp-objective}
  \min\quad
  \sum_{uv\in E^+}w_{uv}x_{uv}
  +\sum_{uv\in E^-}w_{uv}(1-x_{uv}).
\end{equation}
Only the objective coefficients depend on the weights.

For a set $S\subseteq V$, define its cluster cost as
\begin{equation}\label{eq:cluster-cost}
  \cost(S):=\frac12w^+(S,V\setminus S)
          +\frac12w^-(S,S).
\end{equation}
The coverage constraints imply,
for each $uv\in\binom{V}{2}$,
\[
  \sum_{\substack{S:\,u\in S,\ v\notin S}}z_S
  =\sum_{\substack{S:\,v\in S,\ u\notin S}}z_S=x_{uv},
  \qquad
  \sum_{S\supseteq\{u,v\}}z_S=1-x_{uv}.
\]
Thus each positive pair contributes $w_{uv}x_{uv}$ and each negative
pair contributes $w_{uv}(1-x_{uv})$, giving
\begin{equation}\label{eq:cluster-cost-decomp}
  \obj(z):=\sum_{S\subseteq V}\cost(S)z_S
  =\sum_{uv\in E^+}w_{uv}x_{uv}
   +\sum_{uv\in E^-}w_{uv}(1-x_{uv}).
\end{equation}
Setting $z_S=1$ on the parts of a partition gives the same objective
as that partition. Hence the fractional optimum is at most $\OPT$;
our construction guarantees $\obj(z)\le(1+\eps)\OPT$.

\paragraph{Rounding cluster LP.}
We refer to the original construction (Algorithm 1 of \cite{cao2024clusterlp}), which also works on general nonnegative weights.
\begin{algorithm}
\caption{Cluster-Based Rounding}
\label{alg:rounding}
    \begin{algorithmic}[1]
        \Require Weighted correlation clustering instance $(V,\binom{V}{2}=E^+\uplus E^-,w)$, \newline a feasible solution $\{z_S\}$ for cluster LP
        \Ensure A clustering $\Ccal$
        \State $\Ccal\gets \emptyset,\,V'\gets V$
        \While{$V'\neq\emptyset$}
            \State Sample $S$ with probability proportional to $z_S$
            \If{$V'\cap S\neq \emptyset$} $\Ccal\gets \Ccal\cup\{V'\cap S\},\,V'\gets V'\setminus S$
            \EndIf
        \EndWhile
        \State \Return $\Ccal$
    \end{algorithmic}
\end{algorithm}
\begin{lemma}[Weighted factor-$2$ cluster-LP rounding]\label{lem:rounding}
\Cref{alg:rounding} produces a clustering $\Ccal$ with expected objective value at most $2\obj(z)$.
\end{lemma}
\begin{proof}
For every pair of distinct vertices $u,v$, the probability of $u$ and $v$ being separated is $\frac{\sum_{|S\cap\{u,v\}|=1}z_S}{\sum_{S\cap\{u,v\}\neq \emptyset}z_S}=\frac{2x_{uv}}{1+x_{uv}}$, while the probability that $u$ and $v$ are in the same cluster is $1-\frac{2x_{uv}}{1+x_{uv}}=\frac{1-x_{uv}}{1+x_{uv}}$.
Therefore,
\begin{align*}\E[\obj(\Ccal)]&=\sum_{uv\in E^+}w_{uv}\Prb[u,v\text{ separated}]+\sum_{uv\in E^-}w_{uv}\Prb[u,v\text{ joined}]\\
&=\sum_{uv\in E^+}w_{uv}\frac{2x_{uv}}{1+x_{uv}}+\sum_{uv\in E^-}w_{uv}\frac{1-x_{uv}}{1+x_{uv}}\\
&\leq \sum_{uv\in E^+}{2w_{uv}x_{uv}}+\sum_{uv\in E^-}2w_{uv}({1-x_{uv}})=2\obj(z).\qedhere
\end{align*}

\end{proof}

By Markov's inequality and independent repetition, taking the cheapest
of $O(\eta^{-1}\log n)$ runs gives cost at most
$(2+\eta)\obj(z)$ with high probability. If $\obj(z)=0$, the
cost is zero almost surely. \Cref{sec:construction} specifies the
accuracy parameters for \Cref{thm:main}.

\section{Spectral preclustering}\label{sec:precluster}

We construct atoms and admissible edges by proving, in order, optimum
preservation, a regular witness, a Ptolemy degree--volume bound, and
a random-walk certificate.
For unordered vertex pairs $F$, write $w(F)=\sum_{uv\in F}w_{uv}$.

\paragraph{Parameters.}
Instances with at most one vertex are immediate, so assume $n\ge2$.
Fix an absolute constant $\beta$ satisfying
\begin{equation}\label{eq:beta-range}
  0<\beta\le\frac{2\sqrt7-5}{6},
\end{equation}
and set $\rho:=20/3$. The rational choice $\beta=1/32$ suffices.
In particular, $\beta<1/10<1/4$. Define the known constant
\begin{equation}\label{eq:alpha-explicit}
  \alpha:=\rho\left(1+\frac{20}{\beta^2}\right).
\end{equation}
Throughout this section, $\eps\in(0,1]$ is fixed.

\paragraph{Initialization.}
The algorithm of Fan, Lee, and Lee~\cite{fan2025pseudometric} produces
a clustering $\widehat\Ccal$ with
$\E[\obj(\widehat\Ccal)]\le(10/3)\OPT$.
Run this algorithm independently
$O(\log n)$ times and choose a cheapest
output $\Ccal^0$. By Markov's inequality and independent repetition,
$\obj(\Ccal^0)\le\rho\OPT$ holds with high probability when $\OPT>0$.
If $\OPT=0$, every run has zero cost almost surely, which is already optimal.
Henceforth, we assume $\OPT>0$ for the remaining sections and $\obj(\Ccal^0)\leq \rho\OPT$.

\begin{theorem}[Spectral preclustering]
\label{thm:preclustering}
For every fixed $\eps\in(0,1]$,
a deterministic
polynomial-time algorithm constructs an atom partition $\Kcal$ and
an inter-atom admissible edge set $E_{\mathrm{adm}}$ such that:
\begin{enumerate}[label=(\roman*)]
\item $\obj(\Kcal)\le\alpha\OPT$;
\item there exists an atom-respecting clustering $\Ccal_1^\star$
      of cost at most $(1+\eps)\OPT$;
\item $\Ccal_1^\star$ is witnessed by $(\Kcal,E_{\text{adm}})$;
\item the total admissible edge weight satisfies
\[
  w(E_{\mathrm{adm}})
  \le64\alpha\left(\frac{2\alpha}{\eps}\right)^{2t_0+1}\OPT,
\]
where $t_0=245$ is the absolute walk-length constant justified below.
\end{enumerate}
\end{theorem}

\subsection{Constructing atoms $\Kcal$}

For $u\in C\in\Ccal^0$, let
\[
  \cost_u(C):=w(u,N_u^+\triangle C)
  =w^+(u,V\setminus C)+w^-(u,C)
\]
be the disagreement weight incident to $u$ in the initial clustering.
We reserve $d_C(K)$ for positive inter-atom degrees later in the proof.
The vertex-wise objective identity gives
\begin{equation}\label{eq:s3-initial-error-sum}
  \sum_{C\in\Ccal^0}\sum_{u\in C}\cost_u(C)
  =\sum_{C\in\Ccal^0}2\cost(C)=2\obj(\Ccal^0).
\end{equation}
\Cref{alg:atoms} isolates vertices with excessive incident error; it
splits an entire cluster when their total error is too large. This is
the weighted analogue of the atom construction in
\cite{cohenaddad2023preclustering,cao2024clusterlp}.

\begin{algorithm}[htbp]
\caption{ConstructAtoms}\label{alg:atoms}
\begin{algorithmic}[1]
\Require An initial clustering $\Ccal^0$ and the fixed constant $\beta$
\Ensure A partition $\Kcal$ into atoms
\State $\Kcal\gets\emptyset$
\ForAll{$C\in\Ccal^0$}
  \State $M_C\gets\{u\in C:\cost_u(C)>\frac\beta2w(u,C)\}$
  \If{$\sum_{u\in M_C}\cost_u(C)>\frac{\beta^2}{20}w(C,C)$}
    \State Add $\{u\}$ to $\Kcal$ for every $u\in C$
  \Else
    \State Add $\{u\}$ to $\Kcal$ for every $u\in M_C$
    \If{$C\setminus M_C\ne\emptyset$}
      \State Add $C\setminus M_C$ to $\Kcal$
    \EndIf
  \EndIf
\EndFor
\State \Return $\Kcal$
\end{algorithmic}
\end{algorithm}

Call an initial cluster \emph{marked} if the complete-splitting test
in \Cref{alg:atoms} succeeds. A nonsingleton retained set
$K=C\setminus M_C$ is a \emph{core atom}; every other atom is a
singleton.

\begin{lemma}[Constant-cost atomization; (i) of \Cref{thm:preclustering}]\label{lem:atom-cost}
$\obj(\Kcal)\leq \alpha\OPT$.
\end{lemma}
\begin{proof}
Only pairs separated within an initial cluster change status. Negative
pairs can only reduce the cost, so it suffices to bound newly cut positive
weight.

For an unmarked $C$, every newly cut pair meets $M_C$, so
\[
  \text{new positive cost in }C
  \le\sum_{u\in M_C}w(u,C)
  \le\frac2\beta\sum_{u\in M_C}\cost_u(C).
\]
This also holds for $M_C=\emptyset$. For a marked cluster, the test gives
\[
  \text{new positive cost in }C
  \le\frac12w(C,C)
  <\frac{10}{\beta^2}\sum_{u\in M_C}\cost_u(C).
\]
Since $2/\beta\le10/\beta^2$, \eqref{eq:s3-initial-error-sum} gives
\begin{align*}
  \obj(\Kcal)-\obj(\Ccal^0)
  &\le\frac{10}{\beta^2}
       \sum_{C\in\Ccal^0}\sum_{u\in M_C}\cost_u(C)\\
  &\le\frac{10}{\beta^2}
       \sum_{C\in\Ccal^0}\sum_{u\in C}\cost_u(C)
   =\frac{20}{\beta^2}\obj(\Ccal^0).
\end{align*}
Applying $\obj(\Ccal^0)\le\rho\OPT$ concludes the proof.
\end{proof}

\begin{lemma}[Atom stability]\label{lem:atom-stability}
If $K=C\setminus M_C$ is a core atom, then
\begin{equation}\label{eq:core-size}
  |C\setminus K|=|M_C|\le\frac\beta5|C|.
\end{equation}
Consequently, for any $u\in K$ with $w(u,K)>0$,
\begin{equation}\label{eq:atom-stability}
  \cost_u(K)<\beta w(u,K);
\end{equation}
If $w(u,K)=0$, then $\cost_u(K)=0$.
\end{lemma}
\begin{proof}
The complete-splitting test failed for $C$, so
\begin{equation}\label{eq:s3-unmarked-cluster}
  \sum_{u\in M_C}\cost_u(C)
  \le\frac{\beta^2}{20}w(C,C);
\end{equation}
If $w(C,C)=0$, any marked vertex would have positive error,
contradicting \eqref{eq:s3-unmarked-cluster}; hence $M_C=\emptyset$.
Suppose $w(C,C)>0$. By the set-wise triangle inequality, for every $u\in C$,
\[
  w(u,C)\ge\frac{w(C,C)}{2|C|}.
\]
If $|M_C|>\frac{\beta}{5}|C|$, the strict marking inequality would imply
\[
  \sum_{u\in M_C}\cost_u(C)
  >\frac\beta2\sum_{u\in M_C}w(u,C)
  \ge\frac{\beta|M_C|}{4|C|}w(C,C)
  >\frac{\beta^2}{20}w(C,C),
\]
again contradicting \eqref{eq:s3-unmarked-cluster}, which proves \eqref{eq:core-size}.


Next, let any $u\in K$ be given.
Because $u$ is unmarked,
\begin{equation}\label{eq:s3-core-initial-stability}
  \cost_u(C)
  \le\frac\beta2w(u,C)
  =\frac\beta2\bigl(w(u,K)+w(u,M_C)\bigr).
\end{equation}
Removing $M_C$ changes only the pairs incident to $M_C$, so
\begin{equation}\label{eq:s3-core-error-change}
  \cost_u(K)
  \le\cost_u(C)+w(u,M_C).
\end{equation}
By the set-wise triangle inequality, we have
\begin{equation}\label{eq:uM-triangle}
  w(u,M_C)\le\frac{|M_C|}{|K|}w(u,K)+\frac{w(M_C,K)}{|K|}.
\end{equation}
For a marked vertex $x\in M_C$, $w(x,K)\le w(x,C)<2\cost_x(C)/\beta$.
Using \eqref{eq:s3-unmarked-cluster} and the set-wise triangle inequality gives
\begin{align}
  w(M_C,K)
  &\le\frac2\beta\sum_{x\in M_C}\cost_x(C)
   \le\frac\beta{10}w(C,C)\nonumber\\
  &\le\frac\beta5|C|w(u,C)
   =\frac\beta5|C|\bigl(w(u,K)+w(u,M_C)\bigr).\label{eq:mc_k}
\end{align}
Substituting \eqref{eq:mc_k} into \eqref{eq:uM-triangle} gives
\[
  \left(|K|-\frac\beta5|C|\right)w(u,M_C)
  \le\left(|M_C|+\frac\beta5|C|\right)w(u,K).
\]
By \eqref{eq:core-size}, $|M_C|\le\beta|C|/5$ and
$|K|\ge(1-\beta/5)|C|$. The left coefficient is at least
$(1-2\beta/5)|C|>0$, which yields
\begin{equation}\label{eq:s3-core-removed-weight}
  w(u,M_C)\le\frac{2\beta}{5-2\beta}w(u,K).
\end{equation}
Together with \eqref{eq:s3-core-initial-stability} and \eqref{eq:s3-core-error-change},
this gives
\begin{align*}
  \cost_u(K)
  &\le\left[\frac\beta2\left(1+\frac{2\beta}{5-2\beta}\right)
                +\frac{2\beta}{5-2\beta}\right]w(u,K)\\
  &=\frac{9\beta}{10-4\beta}w(u,K).
\end{align*}
Since $\beta<1/4$, the coefficient is strictly less than $\beta$.
Therefore, $\cost_u(K)<\beta w(u,K)$ if $w(u,K)>0$, and $\cost_u(K)=0$ if $w(u,K)=0$.
\end{proof}



\begin{lemma}[Atom-respecting optimum]\label{lem:atom-respecting}
There exists an optimal clustering $\Ccal^\star$ that respects $\Kcal$.
\end{lemma}
\begin{proof}
We use splitting, a single-vertex move, and extraction of a whole atom.

\paragraph{Zero-diameter atoms.}
If $w(K,K)=0$, then $w(u,K)=0$ for every $u\in K$.
By \Cref{lem:atom-stability}, for a nonsingleton such atom,
$\cost_u(K)=0$.
Thus extracting $K$ creates no positive boundary penalty or negative
internal penalty, and cannot increase the cost. It preserves every
other intact atom by disjointness. Singleton atoms are already intact.

\paragraph{Local optimality for a positive-diameter atom.}
Let $\Ccal$ be an optimal clustering, and suppose it splits an atom
$K$ with $w(K,K)>0$. By the set-wise triangle inequality,
$w(u,K)\ge w(K,K)/(2|K|)>0$ for every $u\in K$. Thus
\begin{equation}\label{eq:s3-exchange-stability}
  \cost_u(K)<\beta w(u,K)
  \qquad(u\in K).
\end{equation}
Splitting a cluster $C$ into $C\cap K$ and $C\setminus K$ changes
its cost by
$w^+(C\cap K,C\setminus K)-w^-(C\cap K,C\setminus K)$.
Optimality therefore implies
\begin{equation}\label{eq:s3-nosplit}
  w^-(C\cap K,C\setminus K)
  \le w^+(C\cap K,C\setminus K).
\end{equation}
Next, moving $u$ from its current cluster $C_u$ to another cluster $C'$
changes the objective by
$\cost_u(C')-\cost_u(C_u)$.
Optimality therefore implies
\begin{equation}\label{eq:s3-nomove}
    \cost_u(C_u)\leq \cost_u(C').
\end{equation}

\paragraph{The fragment inequality.}
Fix a cluster $C$ that splits $K$.
Write $F:=K\cap C$ for the fragment of $K$ inside $C$ and $R:=K\setminus C$ for the remainder, both of which are nonempty.
We show that the total weight inside the fragment is controlled by its weight to the remainder:
\begin{equation}\label{eq:fragment-bound}
  w(F,F)<c_\beta w(F,R),
  \qquad c_\beta:=\frac{4+16\beta}{1-10\beta}.
\end{equation}
Here $1-10\beta>0$. If $w(F,R)=0$, the set-wise triangle inequality in both
directions gives $w(F,F)=w(R,R)=0$, contradicting $w(K,K)>0$.
If $w(F,F)=0$, the claim is immediate. We therefore assume
$w(F,F)>0$.

For any $u\in K$, whether or not $u\in C$, the set-wise triangle
inequality implies
\begin{equation}\label{eq:s3-outside-triangle}
  w(u,C\setminus K)\le\frac{w(F,C\setminus K)}{|F|}
                   +\frac{|C\setminus K|}{|F|}w(u,F).
\end{equation}
To bound the right hand side, we apply the set-wise triangle inequality to obtain
\[
  \frac1{|F|}\le\frac{2w(u,F)}{w(F,F)},
  \qquad
  \frac{|C\setminus K|}{|F|}\le\frac{2w(F,C\setminus K)}{w(F,F)}.
\]
The second inequality also holds for $C\setminus K=\emptyset$, and $w(u,F)>0$.
By \eqref{eq:s3-nosplit} and the sum of
\eqref{eq:s3-exchange-stability} over vertices of $F$,
\[
  w(F,C\setminus K)\le2w^+(F,C\setminus K)
  <2\beta w(F,K)
  =2\beta\bigl(w(F,F)+w(F,R)\bigr).
\]
Substituting these bounds into \eqref{eq:s3-outside-triangle} yields
\[
  w(u,C\setminus K)\le\frac{4w(F,C\setminus K)}{w(F,F)}w(u,F)
  <8\beta\frac{w(F,K)}{w(F,F)}w(u,F).
\]
Choose $u\in R$ minimizing $w(u,R)$, which gives $w(u,R)\le w(R,R)/|R|$, and let $C'$ be its current
cluster. Since $C'\ne C$, all positive edges from $u$ to $F$ are cut
in $C'$. Optimality and the decomposition $C=F\uplus(C\setminus K)$ give
\begin{align*}
  w(u,F)-w^-(u,F)
  &\le\cost_u(C')\\
  &\le\cost_u(C)\\
  &\le\cost_u(K)+w^+(u,R)+w^-(u,C\setminus K)\\
  &\le\cost_u(K)+w(u,R)+w(u,C\setminus K).
\end{align*}
Since $w^-(u,F)\le\cost_u(K)$, it follows that
\[
  w(u,F)\le w(u,R)+w(u,C\setminus K)+2\cost_u(K).
\]
Substituting the preceding outside-weight bound and
$\cost_u(K)<\beta w(u,K)=\beta\bigl(w(u,F)+w(u,R)\bigr)$,
and dividing by $w(u,F)>0$, yields
\begin{equation}\label{eq:s3-combined-cost}
  1-2\beta
  <(1+2\beta)\frac{w(u,R)}{w(u,F)}
    +8\beta\frac{w(F,K)}{w(F,F)}.
\end{equation}
Two applications
of the set-wise triangle inequality then imply
\[
  \frac{w(u,R)}{w(u,F)}
  \le\frac{w(R,R)}{|R|}\frac{2|F|}{w(F,F)}
  \le\frac{4w(F,R)}{w(F,F)}.
\]
Substituting this estimate into \eqref{eq:s3-combined-cost} and
multiplying by $w(F,F)>0$ gives
\begin{align*}
  (1-2\beta)w(F,F)
  &<4(1+2\beta)w(F,R)+8\beta w(F,K)\\
  &=8\beta w(F,F)+(4+16\beta)w(F,R).
\end{align*}
Rearranging yields
\[
  (1-10\beta)w(F,F)<(4+16\beta)w(F,R).
\]
Since $1-10\beta>0$, this proves \eqref{eq:fragment-bound}.

\paragraph{The exact cost of extracting the atom.}
Let $C_1,\ldots,C_\ell$ be the clusters splitting $K$, and write
$F_i:=K\cap C_i$ for each fragment and $R_i:=K\setminus C_i$ for its remainder.
Extract $K$ into a new cluster, leaving each nonempty $C_i\setminus K$ in place.
Only fragment--outside pairs and cross-fragment pairs change status:
\begin{align}
  \obj(\Ccal')-\obj(\Ccal)
  &=\sum_i\bigl(w^+(F_i,C_i\setminus K)-w^-(F_i,C_i\setminus K)\bigr)\notag\\
  &\quad+\frac12\sum_i
              \bigl(w^-(F_i,R_i)-w^+(F_i,R_i)\bigr).
  \label{eq:s3-extraction-change}
\end{align}
The second sum counts each cross-fragment pair twice, hence its factor $1/2$.

Set $X:=\sum_i w(F_i,R_i)>0$. By \eqref{eq:fragment-bound},
\[
  w(K,K)=\sum_i w(F_i,F_i)+X<(1+c_\beta)X.
\]
The positive fragment--outside weights and negative cross-fragment weights
are included in $\sum_{u\in K}\cost_u(K)$.
Thus \eqref{eq:s3-extraction-change} and \eqref{eq:s3-exchange-stability} give
\begin{align}
  \obj(\Ccal')-\obj(\Ccal)
  &\le\sum_{u\in K}\cost_u(K)-\frac X2\notag\\
  &<\beta w(K,K)-\frac X2\notag\\
  &<\bigl(\beta(1+c_\beta)-1/2\bigr)X\le0.
  \label{eq:s3-extraction-budget}
\end{align}
The last inequality follows from the chosen range of $\beta$ in
\eqref{eq:beta-range}, which ensures $\beta(1+c_\beta)\le1/2$.
The preceding inequalities remain strict even at the upper endpoint.
This contradicts optimality. Thus every optimum keeps positive-diameter
atoms intact. Extracting any split zero-diameter atoms preserves the
cost and all other intact atoms, proving the lemma.
\end{proof}
\subsection{Constructing admissible edges $E_{\mathrm{adm}}$}
We first introduce notation for the random walks.
Use the zero-diagonal kernels $\omega,\omega^+$ from
\eqref{eq:block-kernels} with $\mathcal P=\Kcal$.
For a union $X$ of atoms, let
$\Acal_X=\{K\in\Kcal:K\subseteq X\}$. For unions $X,Y$, define
\[
 w^+_{\Kcal}(X,Y)
 :=\omega^+(\Acal_X,\Acal_Y)
 =\sum_{\substack{K\subseteq X,\ L\subseteq Y\\K,L\in\Kcal,\ K\ne L}}
      w^+(K,L).
\]
For disjoint $X,Y$ this equals $w^+(X,Y)$; otherwise it excludes
within-atom weights. The global degree is
\[
 d(K):=\omega^+(K,\Kcal)=w^+_{\Kcal}(K,V)
      =w^+(K,V)-w^+(K,K).
\]
The atom partition cuts every positive pair counted by these degrees,
so
\begin{equation}\label{eq:s3-degree-budget}
 \sum_{K\in\Kcal}d(K)
 =\omega^+(\Kcal,\Kcal)
 =2\sum_{\{K,L\}\in\binom{\Kcal}{2}}w^+(K,L)
 \le2\obj(\Kcal)\le2\alpha\OPT.
\end{equation}
Also, set $\tau:=\eps/(2\alpha)\in(0,1)$. 
\paragraph{A regular near-optimal clustering.}
\label{subsec:nearopt}
The next refinement provides the near-optimal clustering to be witnessed, which also has good spectral properties.

\begin{lemma}[Regular atom-respecting clustering; (ii) of \Cref{thm:preclustering}]\label{lem:regular-clustering}
There exists an atom-respecting clustering $\Ccal_1^\star$ of cost
at most $(1+\eps)\OPT$ such that every cluster $C\in\Ccal_1^\star$
satisfies:
\begin{enumerate}[label=(R\arabic*)]
\item\label{it:cut-lower}
For every nontrivial partition $S\uplus T=C$ into unions of atoms,
\[
  w^+(S,T)>\frac12w(S,T).
\]
\item\label{it:leakage}
For every atom $K\subsetneq C$,
\[
  w^+_{\Kcal}(K,C)>\tau d(K).
\]
\end{enumerate}
\end{lemma}
\begin{proof}
Starting from $\Ccal^\star$, repeatedly split a violating cluster.

If \ref{it:cut-lower} fails for $S\uplus T=C$, replace $C$ by $S$ and
$T$. Only pairs across this cut change status, so the cost change is
\[
  w^+(S,T)-w^-(S,T)=2w^+(S,T)-w(S,T)\le0.
\]
Splitting equality cases ensures the strict condition at termination.

If \ref{it:leakage} fails for an atom $K\subsetneq C$, replace $C$
by $K$ and $C\setminus K$. Its cost increase is at most
\[
  w^+(K,C\setminus K)=w^+_{\Kcal}(K,C)\le\tau d(K).
\]
Charge $\tau d(K)$ to the isolated atom. The process never merges
clusters, so each atom is charged at most once, even if later splits
create new violations. By \eqref{eq:s3-degree-budget}, the total is
\[
  \tau\sum_{K\in\Kcal}d(K)
  \le\frac{\eps}{2\alpha}\,2\alpha\OPT=\eps\OPT.
\]
All other splits have nonpositive cost. Each increases the cluster
count, so at most $|\Kcal|-1$ splits occur. Both conditions hold at
termination, with total cost at most $(1+\eps)\OPT$.
\end{proof}

The algorithm need not find the witness or its violating cuts.
Let $\Kcal_+=\{K:d(K)>0\}$. A zero-degree atom must already be a
separate witness cluster, since its singleton atom cut would otherwise
violate \ref{it:cut-lower}. We exclude these atoms from walk matrices.

\paragraph{A local constant-length certificate.}

Fix a witness cluster $C$ containing at least two atoms, and write
$\Acal_C=\{K\in\Kcal:K\subseteq C\}$. Its positive atom graph has
adjacency kernel $\omega^+|_{\Acal_C\times\Acal_C}$. Define
\[
 d_C(K):=\omega^+(K,\Acal_C)=w^+(K,C\setminus K),
 \qquad
 \vol_C:=\omega^+(\Acal_C,\Acal_C)
        =\sum_{K\in\Acal_C}d_C(K).
\]
For a union $S$ of atoms, put $\vol_C(S)=\sum_{K\subseteq S}d_C(K)$.
All degrees and volumes exclude within-atom weights; in particular,
$\vol_C=w^+(C,C)-\sum_{K\in\Acal_C}w^+(K,K)$.
Condition \ref{it:cut-lower} makes every nontrivial atom cut
positive, so this graph is connected and all its degrees are positive.

\begin{claim}[Conductance]\label{clm:conductance}
For every nontrivial atom partition $S\uplus T=C$,
\[
  \min\{\vol_C(S),\vol_C(T)\}<6w^+(S,T).
\]
Consequently the conductance
\[
  \Phi_C:=\min_{\substack{\emptyset\ne S\subsetneq C\\
                         S\text{ a union of atoms}}}
      \frac{w^+(S,C\setminus S)}
           {\min\{\vol_C(S),\vol_C(C\setminus S)\}}
\]
is greater than $1/6$.
\end{claim}
\begin{proof}
Relabel the two sides so that $|S|\le|T|$, where these cardinalities
count \emph{original vertices}, not atoms. The set-wise triangle inequality gives
\[
  w(S,C)=w(S,S)+w(S,T)
  \le\left(2\frac{|S|}{|T|}+1\right)w(S,T)
  \le3w(S,T).
\]
Because $\vol_C(S)\le w(S,C)$, \ref{it:cut-lower} yields
\[
  \min\{\vol_C(S),\vol_C(T)\}
  \le\vol_C(S)\le w(S,C)
  \le3w(S,T)<6w^+(S,T),
\]
We used cardinality only to choose a side; its positive volume need
not be the smaller one.
\end{proof}

Summing block Ptolemy now gives the degree--volume bound.

\begin{claim}[Pair-volume comparison]\label{clm:pair-volume}
For distinct $A,B\in\Acal_C$,
\begin{equation}\label{eq:pair-volume}
 w(A,B)<16\frac{d_C(A)d_C(B)}{\vol_C}.
\end{equation}
\end{claim}
\begin{proof}
Put
\[
 a:=\omega(A,\Acal_C)=w(A,C\setminus A),\quad
 b:=\omega(B,\Acal_C)=w(B,C\setminus B),\quad
 W:=\omega(\Acal_C,\Acal_C).
\]
Since the zero-diagonal block kernel satisfies the factor-two
Ptolemy inequality, summing \eqref{eq:block-ptolemy} over all
$X,Y\in\Acal_C$ gives
\begin{align}
 w(A,B)W
 &\le2\sum_{X,Y\in\Acal_C}
  \bigl(\omega(A,X)\omega(B,Y)+\omega(A,Y)\omega(B,X)\bigr)
       \notag\\
 &=4ab.\label{eq:external-metric-volume}
\end{align}
\Cref{prop:block-ptolemy} includes repeated block indices.
Apply \ref{it:cut-lower} to the cuts separating $A$ and $B$ from
the rest of $C$. It yields
\[
 a<2d_C(A),\qquad b<2d_C(B).
\]
Also $W\ge\vol_C>0$, since $W$ counts both signs and $\vol_C$
counts only positive inter-atom weights. Therefore
\[
 w(A,B)\vol_C\le w(A,B)W\le4ab<16d_C(A)d_C(B).
\]
Divide by $\vol_C>0$; no positivity of $w(A,B)$ is required.
\end{proof}

 Define the row-stochastic transition matrix
\begin{equation}\label{eq:PC}
  P_C(A,B)=\frac{\omega^+(A,B)}{d_C(A)}
  \qquad(A,B\in\Acal_C).
\end{equation}
Its diagonal is zero and its row sum is
$\omega^+(A,\Acal_C)/d_C(A)=1$.
The stationary distribution of $P_C$ is
\begin{equation}\label{eq:piC}
  \pi_C(A)=\frac{d_C(A)}{\vol_C}.
\end{equation}
Indeed, for distinct $A,B$,
$\pi_C(A)P_C(A,B)=w^+(A,B)/\vol_C=\pi_C(B)P_C(B,A)$.
Thus the walk is reversible. We use row vectors.
For arbitrary real row vectors $f,g$ on $\Acal_C$, write
\[
  \langle f,g\rangle_{\pi_C^{-1}}
  :=\sum_{A\in\Acal_C}\frac{f(A)g(A)}{\pi_C(A)},
  \qquad
  \|f\|_{\pi_C^{-1}}^2:=\langle f,f\rangle_{\pi_C^{-1}}.
\]
For a distribution $q$, $\|q-\pi_C\|_{\pi_C^{-1}}^2$ is its
$\chi^2$ divergence from stationarity.

\begin{claim}[Warm start]\label{clm:warm-start}
For every $A\in\Acal_C$,
\[
  \|\delta_AP_C-\pi_C\|_{\pi_C^{-1}}^2\le15.
\]
\end{claim}
\begin{proof}
For $B\ne A$, \Cref{clm:pair-volume} and $w^+(A,B)\le w(A,B)$ give
\[
  \frac{P_C(A,B)}{\pi_C(B)}
  =\frac{w^+(A,B)\vol_C}{d_C(A)d_C(B)}\le16.
\]
The same bound holds for $B=A$ because $P_C(A,A)=0$.
Since both distributions have mass one,
\[
 \|\delta_AP_C-\pi_C\|_{\pi_C^{-1}}^2
 =\sum_B P_C(A,B)\frac{P_C(A,B)}{\pi_C(B)}-1
 \le16\sum_B P_C(A,B)-1=15.
\]
\end{proof}

Put $P_{C,\mathrm{lazy}}:=(I+P_C)/2$, and fix the absolute integer
\begin{equation}\label{eq:t0}
  t_0:=\left\lceil\frac{\log_{144/143}30}{2}\right\rceil=245.
\end{equation}

\begin{lemma}[Constant-length overlap]\label{lem:s3-overlap}
For $A\in\Acal_C$, let
$q_A:=\delta_AP_CP_{C,\mathrm{lazy}}^{t_0}$. Then
\begin{equation}\label{eq:mixing}
  \|q_A-\pi_C\|_{\pi_C^{-1}}^2
  \le15\left(\frac{143}{144}\right)^{2t_0}\le\frac12.
\end{equation}
For any $A,B\in\Acal_C$,
\begin{equation}\label{eq:overlap-lower}
  \langle q_A,q_B\rangle_{\pi_C^{-1}}\ge\frac12,
\end{equation}
or, equivalently,
\begin{equation}\label{eq:overlap-expand}
  \sum_{K\in\Acal_C}\frac{q_A(K)q_B(K)}{d_C(K)}
  \ge\frac1{2\vol_C}.
\end{equation}
\end{lemma}
\begin{proof}
List the eigenvalues of $P_C$ in decreasing order, where $\lambda_1=1$. Cheeger's inequality~\cite{sinclair1989approximate}
and \Cref{clm:conductance} give
\[
  1-\lambda_2(P_C)\ge\frac{\Phi_C^2}{2}>\frac1{72}.
\]
Lazification maps each eigenvalue $\lambda$ to $(1+\lambda)/2$.
Thus every eigenvalue of $P_{C,\mathrm{lazy}}$ is nonnegative, and
its second largest is less than $143/144$.

Let $\Pi_C:=\operatorname{diag}(\pi_C)$ and
$\mathsf S_C:=\Pi_C^{1/2}P_{C,\mathrm{lazy}}\Pi_C^{-1/2}$.
Reversibility makes $\mathsf S_C$ symmetric. Its unit eigenvector for
eigenvalue one is the column vector with coordinates $\sqrt{\pi_C(K)}$.
If a row vector $h$ has zero total mass, then $h\Pi_C^{-1/2}$ is
orthogonal to this vector. Therefore, for every integer $t\ge0$,
\begin{align}
  \|hP_{C,\mathrm{lazy}}^t\|_{\pi_C^{-1}}^2
  &=\|h\Pi_C^{-1/2}\mathsf S_C^t\|_2^2\notag\\
  &\le\left(\frac{143}{144}\right)^{2t}
       \|h\Pi_C^{-1/2}\|_2^2
   =\left(\frac{143}{144}\right)^{2t}\|h\|_{\pi_C^{-1}}^2.
  \label{eq:s3-spectral-contraction}
\end{align}
With $h=\delta_AP_C-\pi_C$, stationarity gives
$hP_{C,\mathrm{lazy}}^{t_0}=q_A-\pi_C$. Apply \Cref{clm:warm-start}
and $(143/144)^{2t_0}\le1/30$ to obtain \eqref{eq:mixing}.

Both $q_A-\pi_C$ and $q_B-\pi_C$ have mass zero, so their weighted
inner products with $\pi_C$ vanish. Expanding and applying
Cauchy--Schwarz,
\begin{align*}
  \langle q_A,q_B\rangle_{\pi_C^{-1}}
  &=1+\langle q_A-\pi_C,q_B-\pi_C\rangle_{\pi_C^{-1}}\\
  &\ge1-\|q_A-\pi_C\|_{\pi_C^{-1}}
             \|q_B-\pi_C\|_{\pi_C^{-1}}
   \ge\frac12.
\end{align*}
Finally, $\pi_C(K)=d_C(K)/\vol_C$ converts this inequality into
\eqref{eq:overlap-expand}.
\end{proof}

Combining \eqref{eq:pair-volume} and \eqref{eq:overlap-expand} gives the local
certificate
\begin{equation}\label{eq:local-certificate}
  w(A,B)<32d_C(A)d_C(B)
       \sum_{K\in\Acal_C}\frac{q_A(K)q_B(K)}{d_C(K)}
       \qquad(A\ne B).
\end{equation}
The non-lazy step replaces the potentially unbounded initial divergence
$1/\pi_C(A)-1$ by $15$; laziness removes negative eigenvalues.
Thus the constant walk length is independent of stationary masses.

\paragraph{A global, computable certificate.}
\label{subsec:cert_global}
We compare the witness walks with a single computable global kernel.

Let $G$ be the positive inter-atom graph restricted to $\Kcal_+$, and
let $D:=\operatorname{diag}(d(K):K\in\Kcal_+)$. Define
\begin{equation}\label{eq:global-P}
  P(K,L)=\frac{\omega^+(K,L)}{d(K)}
  \quad(K,L\in\Kcal_+),
  \qquad P_{\mathrm{lazy}}=\frac{I+P}{2}.
\end{equation}
Removing zero-degree atoms discards no positive edge, so $P$ is
row-stochastic with zero diagonal.
Set
\[
  Q:=PP_{\mathrm{lazy}}^{t_0},
  \qquad
  \Gamma:=DQD^{-1}Q^\top D.
\]
Every inverse here is on positive-degree atoms only. If
$\Kcal_+=\emptyset$, define $\Gamma$ to be the zero matrix without
forming $P$ or $D^{-1}$. Otherwise, extend $\Gamma$ by zero rows and
columns to all of $\Kcal$.

For $A,B\in\Kcal_+$,
\begin{equation}\label{eq:gamma-collision}
  \Gamma_{A,B}
  =d(A)d(B)\sum_{X\in\Kcal_+}
              \frac{Q(A,X)Q(B,X)}{d(X)}.
\end{equation}
Hence $\Gamma$ is symmetric and entrywise nonnegative; it is also the
Gram matrix $HH^\top$ for $H=DQD^{-1/2}$. The argument below uses
entrywise, not positive-semidefinite, comparison.

\paragraph{Comparing local and global flows.}
Fix again a multi-atom witness cluster $C$. Let
$D_C:=\operatorname{diag}(d_C(K):K\in\Acal_C)$ and
$Q_C:=P_CP_{C,\mathrm{lazy}}^{t_0}$. For its local kernel, write
\[
  \Gamma^C:=D_CQ_CD_C^{-1}Q_C^\top D_C.
\]
Because $Q_C(A,K)=q_A(K)$, \eqref{eq:local-certificate} says
$w(A,B)<32\Gamma^C_{A,B}$ for distinct atoms in $C$.

\begin{lemma}[Local-to-global domination]\label{lem:s3-domination}
For $A,B\in\Acal_C$,
\begin{equation}\label{eq:matrix-domination}
  \bigl[D_CQ_CD_C^{-1}Q_C^\top D_C\bigr]_{A,B}
  \le\tau^{-(2t_0+1)}
       \bigl[DQD^{-1}Q^\top D\bigr]_{A,B}.
\end{equation}
\end{lemma}
\begin{proof}
Property \ref{it:leakage} gives $d_C(K)>\tau d(K)$ for every
$K\in\Acal_C$, while $d_C(K)\le d(K)$ is immediate. For local states,
the first transition has an exact flow identity:
\begin{equation}\label{eq:DP-local}
  (DP)_{A,B}=(D_CP_C)_{A,B}=\omega^+(A,B)
  \qquad(A,B\in\Acal_C).
\end{equation}
This includes the diagonal and costs no degree-retention factor.
For a subsequent lazy step and $A\ne B$,
\[
  (P_{\mathrm{lazy}})_{A,B}
  =\frac{w^+(A,B)}{2d(A)}
  =\frac{d_C(A)}{d(A)}(P_{C,\mathrm{lazy}})_{A,B}
  \ge\tau(P_{C,\mathrm{lazy}})_{A,B}.
\]
For $A=B$, both lazy probabilities are $1/2$, and $\tau<1$ gives
the same inequality. Hence, entrywise on $\Acal_C$,
\begin{equation}\label{eq:Plazy-dom}
  (P_{\mathrm{lazy}})_{A,B}
  \ge\tau(P_{C,\mathrm{lazy}})_{A,B}.
\end{equation}
The degree-retention inequality also implies
\begin{equation}\label{eq:D-dom}
  (D^{-1})_{K,K}=\frac1{d(K)}
  \ge\frac\tau{d_C(K)}
  =\tau(D_C^{-1})_{K,K}.
\end{equation}

To compare products of matrices on different state spaces, we
explicitly restrict global intermediate states to $\Acal_C$.
We claim that, for every integer $j\ge0$ and $A,X\in\Acal_C$,
\begin{equation}\label{eq:s3-flow-induction}
  (DPP_{\mathrm{lazy}}^j)_{A,X}
  \ge\tau^j(D_CP_CP_{C,\mathrm{lazy}}^j)_{A,X}.
\end{equation}
The case $j=0$ is \eqref{eq:DP-local}. If it holds for $j$, then
nonnegativity and \eqref{eq:Plazy-dom} give
\begin{align*}
  (DPP_{\mathrm{lazy}}^{j+1})_{A,X}
  &=\sum_{Y\in\Kcal_+}
        (DPP_{\mathrm{lazy}}^j)_{A,Y}(P_{\mathrm{lazy}})_{Y,X}\\
  &\ge\sum_{Y\in\Acal_C}
        (DPP_{\mathrm{lazy}}^j)_{A,Y}(P_{\mathrm{lazy}})_{Y,X}\\
  &\ge\tau^{j+1}\sum_{Y\in\Acal_C}
        (D_CP_CP_{C,\mathrm{lazy}}^j)_{A,Y}
        (P_{C,\mathrm{lazy}})_{Y,X}\\
  &=\tau^{j+1}(D_CP_CP_{C,\mathrm{lazy}}^{j+1})_{A,X}.
\end{align*}
Thus $(DQ)_{A,X}\ge\tau^{t_0}(D_CQ_C)_{A,X}$.

Finally, expand the global kernel and keep only meeting points in $C$:
\begin{align*}
  \Gamma_{A,B}
  &=\sum_{X\in\Kcal_+}
       \frac{(DQ)_{A,X}(DQ)_{B,X}}{d(X)}\\
  &\ge\sum_{X\in\Acal_C}
       \frac{(DQ)_{A,X}(DQ)_{B,X}}{d(X)}\\
  &\ge\tau^{2t_0+1}\sum_{X\in\Acal_C}
       \frac{(D_CQ_C)_{A,X}(D_CQ_C)_{B,X}}{d_C(X)}
   =\tau^{2t_0+1}\Gamma^C_{A,B}.
\end{align*}
The exponent counts two sets of $t_0$ lazy transitions and one
inverse-degree factor. No global mixing assumption is used.
\end{proof}

It follows from \eqref{eq:local-certificate} and \eqref{eq:matrix-domination} that
all distinct atoms sharing a witness cluster satisfy
\begin{equation}\label{eq:s3-witness-certificate}
  w(A,B)<32\tau^{-(2t_0+1)}\Gamma_{A,B}.
\end{equation}

\begin{definition}[Admissible edges]\label{def:eadm}
For vertices $u,v$ in distinct atoms $K_u,K_v$, declare
$uv\in E_{\mathrm{adm}}$ if
\begin{equation}\label{eq:eadm-def}
  w(K_u,K_v)<32\tau^{-(2t_0+1)}\Gamma_{K_u,K_v}.
\end{equation}
Pairs inside the same atom are not included in $E_{\mathrm{adm}}$.
\end{definition}

Admissibility is symmetric and constant on each atom pair. Zero-degree
atoms have zero kernel rows and pass no strict test; they are separate
witness clusters, so no required pair is lost. The test never divides
by an input weight.

\paragraph{The total mass of the global kernel.}

\begin{lemma}[Kernel mass]\label{lem:s3-kernel-mass}
On $\Kcal_+$, $\Gamma\one=D\one$. Consequently,
\begin{equation}\label{eq:s3-kernel-mass}
  \sum_{A,B\in\Kcal}\Gamma_{A,B}
  =\sum_{K\in\Kcal}d(K)
  \le2\obj(\Kcal)\le2\alpha\OPT.
\end{equation}
The same total-mass statement holds when $\Kcal_+=\emptyset$.
\end{lemma}
\begin{proof}
If $\Kcal_+$ is empty, both sums are zero. Otherwise let
$\mathbf d:=D\one$. Undirectedness gives detailed balance
$d(A)P(A,B)=d(B)P(B,A)$, and summing over $A$ yields
$P^\top\mathbf d=\mathbf d$. Lazification preserves this identity,
and so does multiplication of transition matrices preserving the
same degree vector. Therefore
\[
  Q\one=\one,
  \qquad Q^\top\mathbf d=\mathbf d.
\]
These identities hold even on a disconnected graph and give
\[
  \Gamma\one
  =DQD^{-1}Q^\top D\one
  =DQD^{-1}\mathbf d
  =DQ\one=D\one.
\]
Summing the row sums and including the zero rows proves the equality
in \eqref{eq:s3-kernel-mass}. Its inequalities are precisely
\eqref{eq:s3-degree-budget}.
\end{proof}

\paragraph{Completion of the construction.}
The algorithm computes $\Kcal$ using \Cref{alg:atoms}, forms the
positive inter-atom weights and degrees, constructs $Q$ and $\Gamma$
on $\Kcal_+$, and applies \eqref{eq:eadm-def} to each unordered pair
of atoms. It does not compute any of the witness objects introduced
in the proofs.

\begin{proof}[Proof of \Cref{thm:preclustering}]
Assertion (i) follows from \Cref{lem:atom-cost} and the stipulated
initial bound. For (ii), \Cref{lem:atom-respecting} supplies an
atom-respecting optimum and \Cref{lem:regular-clustering} supplies
the required near-optimal witness. For (iii), a multi-atom witness
cluster has positive local degrees, and
\eqref{eq:s3-witness-certificate} and \eqref{eq:eadm-def} admit every vertex pair
between any two of its atoms. A single-atom witness cluster has no
inter-atom pair to check.

For (iv), let $\mathcal P_{\mathrm{adm}}\subseteq\binom{\Kcal}{2}$
be the unordered atom pairs passing the test. Blockwise admissibility
and disjointness of the atoms give the exact equality
\[
  w(E_{\mathrm{adm}})
  =\sum_{\{A,B\}\in\mathcal P_{\mathrm{adm}}}w(A,B).
\]
If this family is empty, the desired bound is immediate. Otherwise,
sum the test over this family. Since all kernel entries are nonnegative,
we may bound the sum over admissible unordered atom pairs by the sum
over all ordered atom pairs:
\begin{align*}
  w(E_{\mathrm{adm}})
  &\le32\tau^{-(2t_0+1)}
          \sum_{\{A,B\}\in\mathcal P_{\mathrm{adm}}}\Gamma_{A,B}\\
  &\le32\tau^{-(2t_0+1)}\sum_{A,B\in\Kcal}\Gamma_{A,B}\\
  &\le64\alpha\tau^{-(2t_0+1)}\OPT
   =64\alpha\left(\frac{2\alpha}{\eps}\right)^{2t_0+1}\OPT,
\end{align*}
where \Cref{lem:s3-kernel-mass} gives the third line.

Atomization, aggregation, and the constant number of matrix products
use polynomially many arithmetic operations. The matrices have at most
$n$ rows. For rational weights, rational accuracy, and fixed rational
$\beta$, all sums, products, divisions by positive degrees, and tests
have polynomial bit complexity.
The construction uses neither $\OPT$
nor the witness.
\end{proof}

For the remaining sections, fix the explicit quantity
\begin{equation}\label{eq:R-def}
  R_\eps:=\max\left\{1,\alpha,
    64\alpha\left(\frac{2\alpha}{\eps}\right)^{2t_0+1}\right\}.
\end{equation}
Thus $R_\eps=O(\eps^{-491})$ is explicit and
$w(E_{\mathrm{adm}})\le R_\eps\OPT$.


\section{Solving the bounded sub-cluster LP}\label{sec:boundedlp}

Following the unweighted cluster-LP construction~\cite{cao2024clusterlp},
we retain size-indexed local moments. Let $(\Kcal,E_{\mathrm{adm}})$ be the preclustering constructed in the previous section.

    
Recall the explicit constant $R_\eps$ from \eqref{eq:R-def}, for which
$w(E_{\mathrm{adm}})\le R_\eps\OPT$.
Set
\begin{equation}\label{eq:r-def}
  r:=2+\left\lceil\frac{2\ln2\,R_\eps^2}{\eps^2}\right\rceil.
\end{equation}
Then $r\ge3$ and $r=O(\eps^{-984})$. The extra $\eps^{-2}$ beyond
$R_\eps^2$ pays for the desired sampling accuracy, as quantified below.
All parameters are known without evaluating $\OPT$.

For $s\in[n]$ and $S\subseteq V$ such that $|S|\le r$, introduce $y_S^s$ and write
\[
  y_S:=\sum_{s=1}^n y_S^s,
  \qquad Su:=S\cup\{u\},
  \qquad Suv:=S\cup\{u,v\}.
\]
When the moments arise from a partition $\Ccal$, their intended
meaning is
\begin{equation}\label{eq:bdd-integral-moments}
  y_S^s
  =\sum_{C\in\Ccal:\,|C|=s}\mathbf 1[S\subseteq C]=\sum_{C\supseteq S,\,|C|=s}z_C.
\end{equation}
Thus $s$ indexes cluster size, whereas $|S|$ is moment order bounded by $r$. Each size-$s$ cluster containing $S$ is counted
$s$ times when a vertex is adjoined, giving
\begin{equation}\label{eq:bdd-size-motivation}
  \sum_{u\in V}y_{Su}^s=s\,y_S^s.
\end{equation}
This equation plays a role as the one-step seed transition process, conditioned on the size index $s$ and the current seed $S$.

The bounded sub-cluster LP is
\begin{align}
 \min\quad
 &\sum_{uv\in E^+}w_{uv}(1-y_{uv})
       +\sum_{uv\in E^-}w_{uv}y_{uv},\label{eq:bdd-obj}\\
 \text{subject to}\quad
 &y_u=1 &&\forall u\in V,\label{eq:bdd-single}\\
 &\frac1s\sum_{u\in V}y_{Su}^s=y_S^s
       &&\forall s\in[n],\ |S|\le r-1,\label{eq:bdd-size}\\
 &y_S^s\ge0
       &&\forall s\in[n],\ |S|\le r,\label{eq:bdd-nonneg}\\
 &y_{uv}=1
       &&\forall uv,\ K_u= K_v,\label{eq:bdd-atom}\\
 &y_{uv}=0
       &&\forall uv,\ K_u\neq K_v,\notag\\[-0.4ex]
 &&&\hspace{1.7em}uv\notin E_{\mathrm{adm}},\label{eq:bdd-adm}\\
 &0\le\sum_{T'\subseteq T}(-1)^{|T'|}y_{S\cup T'}^s\le y_S^s
       &&{\forall s\in[n],\ S,T\subseteq V,}\notag\\[-3ex]
       &&&\hspace{1.7em}|S\uplus T|\leq r.\label{eq:bdd-SA}
\end{align}
The pair constraints concern distinct vertices: within-atom pairs use
\eqref{eq:bdd-atom}, and only inter-atom pairs use \eqref{eq:bdd-adm}.
The fixed-size identities normalize seed transitions; inclusion--exclusion
supplies the local probability laws summarized below.

The LP has $n\sum_{j\le\min\{r,n\}}\binom nj=n^{O(r)}$ variables
and $O(n\sum_{j\le\min\{r,n\}}\binom nj 2^r)n^{O(r)}$ constraints, dominated by inclusion--exclusion constraints.
Therefore, the LP can be solved in $n^{O(r)}$ times a polynomial in the
input encoding length.

\begin{lemma}[Feasibility and objective bound]\label{lem:bdd-feasible}
The bounded sub-cluster
LP is feasible, and its optimum is at most $(1+\eps)\OPT$.
\end{lemma}
\begin{proof}
Use the witness $\Ccal_1^\star$ of \Cref{thm:preclustering} in
\eqref{eq:bdd-integral-moments}. Every vertex belongs to one cluster,
and each size-$s$ cluster containing $S$ contributes exactly $s$ terms
to $\sum_u y_{Su}^s$, including $u\in S$. The witness preserves atoms
and contains only admissible inter-atom edges, so the pair constraints
hold. For disjoint $S,T$, the inclusion--exclusion sum counts size-$s$
witness clusters containing $S$ and avoiding $T$; it lies in $[0,y_S^s]$.
Finally, the objective is $\obj(\Ccal_1^\star)\le(1+\eps)\OPT$.
\end{proof}

\paragraph{Probabilistic interpretation.}
The sampler will use ratios of moments as conditional inclusion
probabilities. The inclusion--exclusion constraints justify this:
for each $s$ with $y_\emptyset^s>0$, the ratios
$y_S^s/y_\emptyset^s$ give consistent inclusion probabilities on
sets of at most $r$ vertices. In particular, $y_{Uv}^s/y_U^s$ is
the conditional probability of including $v$ given inclusion of $U$,
whenever $y_U^s>0$ and $|U\cup\{v\}|\le r$.
This interpretation requires only local distributions; a joint
distribution on all vertices is not assumed.
The moments are monotone:
\begin{equation}\label{eq:bdd-monotonicity}
  0\le y_{S'}^s\le y_S^s\le y_\emptyset^s
  \qquad(S\subseteq S',\ |S'|\le r).
\end{equation}
The fixed-size identities also force $y_S^s=0$ whenever
$s<|S|\le r$. Proofs of these facts appear in
\Cref{app:bdd-local-laws}.

\section{Sampling one fractional cluster}\label{sec:sampling}

Let $\{y_S^s\}$ be an optimal solution of the bounded sub-cluster LP from the previous section. We turn its local moments into an actual set, proving exact
singleton marginals, a row-wise pair-error bound, and a weighted cost
bound. By \Cref{lem:bdd-feasible,eq:R-def},
\begin{equation}\label{eq:sampling-input-bounds}
  \sum_{uv\in E^+}w_{uv}(1-y_{uv})
  +\sum_{uv\in E^-}w_{uv}y_{uv}
  \le(1+\eps)\OPT,
  \qquad w(E_{\mathrm{adm}})\le R_\eps\OPT.
\end{equation}
Write $y_\emptyset:=\sum_s y_\emptyset^s$. Taking $S=\emptyset$
in \eqref{eq:bdd-size} and summing over $s$ gives
\begin{equation}\label{eq:bdd-yempty-bounds}
  \sum_{s=1}^n s\,y_\emptyset^s
    =\sum_{u\in V}y_u=n,
  \qquad 1\le y_\emptyset\le n.
\end{equation}
Thus $y_\emptyset^s/y_\emptyset$ defines a distribution on size indices.
We evaluate moment ratios only on positive-mass prefixes.

\subsection{The sampler and its marginals}

Rounding atoms independently can lose the pairwise correlations
recorded by the LP. We first sample seed vertices and condition the
moments on their inclusion, then make the independent atom decisions
using these conditional moments~\cite{raghavendra2012cardinality}.
The analysis controls the remaining pairwise error averaged over seed
lengths. Accordingly, \Cref{alg:sample} chooses the stopping length
uniformly at random, independently of the size index and seed
randomness.

\begin{algorithm}[htbp]
\caption{SampleOneCluster}\label{alg:sample}
\begin{algorithmic}[1]
\Require An optimal bounded sub-cluster LP solution $\{y_S^s\}$
\Ensure A union $C$ of atoms
\State Sample $s\in[n]$ with probability $y_\emptyset^s/y_\emptyset$
\State Sample $T$ uniformly from $\{1,\ldots,r-2\}$
\State $U\gets\emptyset$
\For{$t=1,\ldots,T$}
  \State Sample $u_t\in V$ with probability $y_{Uu_t}^s/(s y_U^s)$
  \State $U\gets U\cup\{u_t\}$
\EndFor
\State $C\gets\emptyset$
\ForAll{$K\in\Kcal$}
  \State Choose a representative $v\in K$
  \State Independently include all of $K$ in $C$ with probability
         $y_{Uv}^s/y_U^s$
\EndFor
\State \Return $C$
\end{algorithmic}
\end{algorithm}

Let $U_t=\{u_1,\ldots,u_t\}$, with $U_0=\emptyset$.
Seeds may repeat: each draw advances $t$, while $U_t$ records only the distinct vertices drawn.
The transition rule is normalized by the fixed-size identity:
\begin{equation}\label{eq:bdd-seed-kernel}
 \sum_{u\in V}\frac{y_{Uu}^s}{s y_U^s}=1,
 \qquad |U|\le r-1,\quad y_U^s>0.
\end{equation}
Positive-probability transitions keep denominators positive. Zero-mass
sizes are never sampled, and conditionals on null events can be arbitrary.
The probability of choosing an existing seed vertex is $|U|/s\le1$ by
\Cref{lem:bdd-size-support}.
\Cref{lem:bdd-atom-propagation} makes inclusion independent of the chosen
representative and forces every seed atom to be included.

The index $s$ specifies the size in the LP, but independent atom
decisions can produce a set of a different size. Conditional on $s$
and the final seed set $U$, the expected size is
\[
 \sum_{K\in\Kcal}|K|\frac{y_{Uv_K}^s}{y_U^s}
 =\sum_{v\in V}\frac{y_{Uv}^s}{y_U^s}=s,
\]
where $v_K\in K$.

For a vertex $a$, write $K_a$ for its atom and define the
admissible neighborhood
\begin{equation}\label{eq:sampling-neighborhood}
 N(a):=\{u\in V\setminus K_a:au\in E_{\mathrm{adm}}\}.
\end{equation}
Blockwise admissibility makes $N(a)$ a union of atoms.

\begin{lemma}[First-seed support]\label{lem:sampling-seed-support}
Condition on an index $s$ and a first seed $u_1=a$ of positive
probability. Then, almost surely,
\[
  K_a\subseteq C\subseteq K_a\cup N(a).
\]
Moreover, the sampler preserves atoms as indivisible units. It is
\emph{not} asserted that every two distinct atoms contained in
$C$ are mutually admissible.
\end{lemma}
\begin{proof}
Since $a\in U_t$, conditional atom consistency gives inclusion
probability one to $K_a$. For $av\notin E_{\mathrm{adm}}$ in distinct
atoms, \Cref{lem:bdd-atom-propagation}(ii) gives $y_{U_tv}^s=0$.
Thus no atom outside $K_a\cup N(a)$ is selected. Whole-atom inclusion
is explicit in the algorithm.
\end{proof}

 \paragraph{Exact one-vertex marginals.}

Fix $s$ with $y_\emptyset^s>0$ and run the seed chain to a deterministic
horizon $h\le r-1$. Let $\mathcal F_t=(u_1,\ldots,u_t)$ be the
ordered history. For a fixed vertex $v$, attach a Bernoulli reference
variable with law
\begin{equation}\label{eq:terminal-indicator}
 \Prb[B_v^{(h)}=1\mid s,\mathcal F_h]
   :=\frac{y_{U_hv}^s}{y_{U_h}^s}.
\end{equation}
This uses moments of order at most $h+1\le r$, and makes no assumption
of a simultaneous global realization for all vertices.

\begin{claim}[Conditional one-vertex marginal]
\label{clm:conditional-marginal}
For every deterministic $0\le t\le h\le r-1$,
\begin{equation}\label{eq:conditional-one-vertex}
 \Prb[B_v^{(h)}=1\mid s,\mathcal F_t]
   =\frac{y_{U_tv}^s}{y_{U_t}^s}.
\end{equation}
For the actual sampler, condition on $T=h\le r-2$. Then, for
$0\le t\le h$,
\[
 \Prb[v\in C\mid s,\mathcal F_t,T=h]
   =\frac{y_{U_tv}^s}{y_{U_t}^s}.
\]
\end{claim}
\begin{proof}
Fix $s$ and a positive-mass prefix $U_t$. For any ordered
continuation $(u_{t+1},\ldots,u_h)$, the transition
probabilities telescope:
\begin{equation}\label{eq:seed-telescope}
 \Prb[u_{t+1},\ldots,u_h\mid s,\mathcal F_t]
   =\prod_{j=t}^{h-1}\frac{y_{U_{j+1}}^s}{s y_{U_j}^s}
   =\frac{y_{U_h}^s}{s^{h-t}y_{U_t}^s}.
\end{equation}
Impossible continuations have zero terminal moment by monotonicity;
repetitions satisfy $U_ju_{j+1}=U_{j+1}$. Multiplying by
\eqref{eq:terminal-indicator} and summing gives
\begin{align*}
 \Prb[B_v^{(h)}=1\mid s,\mathcal F_t]
 &=\frac{1}{s^{h-t}y_{U_t}^s}
   \sum_{(u_{t+1},\ldots,u_h)\in V^{h-t}}y_{U_hv}^s\\
 &=\frac{y_{U_tv}^s}{y_{U_t}^s}.
\end{align*}
Eliminate seeds in reverse order using \eqref{eq:bdd-size}; each base
set has size at most $h\le r-1$. The actual atom decision has the same
terminal marginal by \Cref{lem:bdd-atom-propagation}, so the calculation
also proves the second assertion.
\end{proof}

\begin{lemma}[Exact one-vertex marginal]\label{lem:one-marginal}
Every vertex $v\in V$ satisfies
\[
 \Prb[v\in C]=\frac1{y_\emptyset}.
\]
\end{lemma}
\begin{proof}
By \Cref{clm:conditional-marginal} at $t=0$,
$\Prb[v\in C\mid s,T=h]=y_v^s/y_\emptyset^s$ for every allowed $h$.
Averaging over $T$ and $s$ gives
\[
 \Prb[v\in C]
 =\sum_{s:y_\emptyset^s>0}
   \frac{y_\emptyset^s}{y_\emptyset}
   \frac{y_v^s}{y_\emptyset^s}
 =\frac{y_v}{y_\emptyset}
 =\frac1{y_\emptyset},
\]
where $y_v=1$.
\end{proof}

\subsection{Bounding the weighted sampling cost}

We use one reference indicator at horizon $H=r-1$ to control residual
correlation at every earlier prefix.

Fix $s$ and $a$ with $y_a^s>0$, and condition throughout the
conditional pair-error argument on $u_1=a$. For $1\le t\le r-2$, the same
fixed-size calculation as above, this time with two additional
vertices, gives
\begin{equation}\label{eq:pair-telescoping}
 \E\left[\frac{y_{U_tvw}^s}{y_{U_t}^s}
          \,\middle|\,s,u_1=a\right]
    =\frac{y_{avw}^s}{y_a^s}.
\end{equation}
Indeed, an ordered continuation $(u_2,\ldots,u_t)$ has
probability $y_{U_t}^s/(s^{t-1}y_a^s)$, so the expectation
reduces to
\[
 \frac1{s^{t-1}y_a^s}
 \sum_{(u_2,\ldots,u_t)\in V^{t-1}}y_{U_tvw}^s.
\]
Eliminating the $t-1$ seed indices by \eqref{eq:bdd-size} proves the
identity using moments of order at most $t+2\le r$.

For a positive-mass prefix, define the singleton marginals whenever
$t\le r-1$ and the pair marginal whenever $t\le r-2$ by
\begin{equation}\label{eq:conditional-prefix-moments}
 p_v(U_t):=\frac{y_{U_tv}^s}{y_{U_t}^s},\qquad
 p_w(U_t):=\frac{y_{U_tw}^s}{y_{U_t}^s},\qquad
 p_{vw}(U_t):=\frac{y_{U_tvw}^s}{y_{U_t}^s}.
\end{equation}
Only a singleton marginal is needed at $H=r-1$; no $(r+1)$-order
pair moment is used.

\begin{lemma}[A common reference experiment]
\label{lem:fixed-horizon-reference}
Let $H=r-1$. Starting from $u_1=a$, run the seed chain through
time $H$, and for a fixed vertex $v$ draw a Bernoulli variable
$B_v$ with
$\Prb[B_v=1\mid s,\mathcal F_H]=p_v(U_H)$.
This joint experiment is well-defined using moments of order at
most $r$. For each $1\le t\le H-1$,
\begin{align}
 \Prb[B_v=1\mid s,\mathcal F_t]&=p_v(U_t),\notag\\
 \Prb[u_{t+1}=w\mid s,\mathcal F_t]&=\frac{p_w(U_t)}s,\notag\\
 \Prb[B_v=1,u_{t+1}=w\mid s,\mathcal F_t]
   &=\frac{p_{vw}(U_t)}s.
   \label{eq:reference-joint}
\end{align}
The construction is made separately for each fixed $v$; no
joint distribution of terminal indicators for different vertices
is required.
\end{lemma}
\begin{proof}
Write $\mathbf u=(u_2,\ldots,u_H)$ for the continuation of the
seed trajectory. For a deterministic sequence
$\mathbf w=(w_2,\ldots,w_H)$, let
$U(\mathbf w)=\{a,w_2,\ldots,w_H\}$. Then
\begin{align*}
 \Prb[\mathbf u=\mathbf w,B_v=1\mid s,u_1=a]
  &=\frac{y_{U(\mathbf w)v}^s}{s^{H-1}y_a^s},\\
 \Prb[\mathbf u=\mathbf w,B_v=0\mid s,u_1=a]
  &=\frac{y_{U(\mathbf w)}^s-y_{U(\mathbf w)v}^s}
          {s^{H-1}y_a^s}.
\end{align*}
Monotonicity makes both expressions nonnegative; summing over sequences
and both indicator values gives one by \eqref{eq:bdd-size}.
The first two identities in \eqref{eq:reference-joint} follow from
\Cref{clm:conditional-marginal} and the transition rule.
For the third, when $p_w(U_t)>0$, condition on $u_{t+1}=w$ and multiply
$y_{U_twv}^s/y_{U_tw}^s$ by $y_{U_tw}^s/(s y_{U_t}^s)$.
When $p_w(U_t)=0$, monotonicity also gives $p_{vw}(U_t)=0$.
All required moments have order at most $H+1=r$.
\end{proof}

\begin{lemma}[Entropy budget for a prefix row]
\label{lem:seed-information-budget}
Fix $s,a$ as above and a vertex $v$. For a positive-probability
history $f_t$, let
\[
 J_t(f_t):=\I(B_v;u_{t+1}\mid s,u_1=a,\mathcal F_t=f_t),
\]
where mutual information is computed with natural logarithms in
the common experiment of
\Cref{lem:fixed-horizon-reference}. Then
\begin{equation}\label{eq:conditional-pinsker}
 \frac1s\sum_{w\in V}
   |p_v(U_t)p_w(U_t)-p_{vw}(U_t)|
 \le\sqrt{\frac{J_t(f_t)}2},
\end{equation}
and
\begin{equation}\label{eq:information-budget}
 \sum_{t=1}^{r-2}\E[J_t(\mathcal F_t)\mid s,u_1=a]
 \le\ln2.
\end{equation}
\end{lemma}
\begin{proof}
Fix a history $f_t$. By \eqref{eq:reference-joint},
under the conditional joint law of $(B_v,u_{t+1})$,
\[
 \Prb[B_v=1,u_{t+1}=w]=\frac{p_{vw}(U_t)}s,
 \quad
 \Prb[B_v=1]\Prb[u_{t+1}=w]
   =\frac{p_v(U_t)p_w(U_t)}s.
\]
The differences for $B_v=0$ are the negatives of the
corresponding differences for $B_v=1$. Since total variation is
one half of the $\ell_1$ distance,
\begin{align*}
 &\bigl\|\mathcal L(B_v,u_{t+1}\mid f_t)
        -\mathcal L(B_v\mid f_t)\otimes
         \mathcal L(u_{t+1}\mid f_t)\bigr\|_{\TV}\\
 &\qquad=\frac1s\sum_{w\in V}
       |p_v(U_t)p_w(U_t)-p_{vw}(U_t)|.
\end{align*}
Pinsker's inequality~\cite{pinsker1964information} bounds this
total variation by $\sqrt{J_t(f_t)/2}$, proving
\eqref{eq:conditional-pinsker}.

For the same $B_v$ at every time, the conditional information chain
rule gives
\begin{align*}
 \sum_{t=1}^{H-1}\E[J_t(\mathcal F_t)\mid s,u_1=a]
 &=\sum_{t=1}^{H-1}
   \I(B_v;u_{t+1}\mid s,u_1=a,u_2,\ldots,u_t)\\
 &=\I(B_v;u_2,\ldots,u_H\mid s,u_1=a)\\
 &\le\Hent(B_v\mid s,u_1=a)\le\ln2.
\end{align*}
The Bernoulli entropy is at most $\ln2$. A different stopped indicator
at each time would not give this common information budget.
\end{proof}

\begin{claim}[Aggregated conditional pair error]
\label{clm:conditional-pair-error}
For every size $s$, first seed $a$ with $y_a^s>0$, and vertex
$v\in V$, the actual output $C$ of \Cref{alg:sample} satisfies
\begin{equation}\label{eq:conditional-pair-error}
 \sum_{w\in V}
 \left|\Prb[v,w\in C\mid s,u_1=a]
       -\frac{y_{avw}^s}{y_a^s}\right|
 \le s\sqrt{\frac{\ln2}{2(r-2)}}.
\end{equation}
\end{claim}
\begin{proof}
All expectations below are conditional on $s,u_1=a$.
Given a stopping time $T=t$ and the complete prefix $U_t$,
the inclusion decisions for \emph{different} atoms are
independent. Thus, for vertices $v,w$ in different atoms,
\[
 \Prb[v,w\in C\mid s,\mathcal F_t,T=t]
   =p_v(U_t)p_w(U_t).
\]
For $v,w$ in the same atom, the actual joint probability and the
local pair moment both equal $p_v(U_t)$, so their discrepancy is
zero. Therefore, in both cases,
\[
 \left|\Prb[v,w\in C\mid s,\mathcal F_t,T=t]-p_{vw}(U_t)\right|
 \le |p_v(U_t)p_w(U_t)-p_{vw}(U_t)|.
\]

By \eqref{eq:pair-telescoping}, the target averages to $y_{avw}^s/y_a^s$
for every $t$. The first-seed law is independent of $T$, so $T$ remains
uniform given $s,a$. Averaging over histories and times yields
\begin{align}
 &\sum_{w\in V}
 \left|\Prb[v,w\in C\mid s,u_1=a]
       -\frac{y_{avw}^s}{y_a^s}\right|\notag\\
 &\quad\le\frac1{r-2}\sum_{t=1}^{r-2}
     \E\left[\sum_{w\in V}
       |p_v(U_t)p_w(U_t)-p_{vw}(U_t)|
          \,\middle|\,s,u_1=a\right].
          \label{eq:row-error-average}
\end{align}

By \eqref{eq:conditional-pinsker}, its inner sum is at most
$s\sqrt{J_t(\mathcal F_t)/2}$. Jensen's inequality, followed
by Cauchy--Schwarz over the stopping times, gives
\begin{align*}
 &\frac{s}{r-2}\sum_{t=1}^{r-2}
      \E\left[\sqrt{\frac{J_t(\mathcal F_t)}2}
               \,\middle|\,s,u_1=a\right]\\
 &\qquad\le
 s\sqrt{\frac{1}{2(r-2)}
      \sum_{t=1}^{r-2}\E[J_t(\mathcal F_t)\mid s,u_1=a]}\\
 &\qquad\le s\sqrt{\frac{\ln2}{2(r-2)}},
\end{align*}
where the final step is \eqref{eq:information-budget}.
\end{proof}

\paragraph{Charging the error to admissible edges.}

We now apply the triangle inequality to the original vertex weights.

For $s,a$ with $y_a^s>0$, define
\begin{equation}\label{eq:conditional-pair-error-symbol}
 e_{uv}^{s,a}:=
  \left|\Prb[u,v\in C\mid s,u_1=a]
                 -\frac{y_{auv}^s}{y_a^s}\right|.
\end{equation}
If $u,v$ lie in one atom, then $e_{uv}^{s,a}=0$ by
conditional atom consistency, exact atom inclusion, and
\eqref{eq:pair-telescoping}. If one endpoint lies in $K_a$,
its inclusion is certain and the other endpoint has the exact
conditional one-vertex marginal, so the error is again zero.
Finally, if an endpoint lies outside $K_a\cup N(a)$, both its
conditional inclusion probability and the corresponding local
moment vanish by
\Cref{lem:bdd-atom-propagation,lem:sampling-seed-support}.
Consequently,
\begin{equation}\label{eq:conditional-error-support}
 e_{uv}^{s,a}\ne0\quad\Longrightarrow\quad u,v\in N(a).
\end{equation}

\begin{lemma}[Total weighted pair error]\label{lem:weighted-error}
The output of \Cref{alg:sample} satisfies
\begin{equation}\label{eq:weighted-error}
 \sum_{uv\in\binom{V}{2}}w_{uv}
 \left|\Prb[u,v\in C]-\frac{y_{uv}}{y_\emptyset}\right|
 \le\frac{\eps\OPT}{y_\emptyset}.
\end{equation}
\end{lemma}
\begin{proof}
Fix $s,a$ as above. By \eqref{eq:conditional-error-support},
only pairs with both endpoints in $N(a)$ contribute.
Using the triangle inequality through the first seed,
$w_{uv}\le w_{ua}+w_{av}$, and symmetry of the errors, we have
\begin{align}
 \sum_{uv\in\binom{V}{2}}w_{uv}e_{uv}^{s,a}
 &=\frac12\sum_{u,v\in N(a)}w_{uv}e_{uv}^{s,a}\notag\\
 &\le\frac12\sum_{u,v\in N(a)}
                (w_{ua}+w_{va})e_{uv}^{s,a}\notag\\
 &=\sum_{u\in N(a)}w_{ua}\sum_{v\in N(a)}e_{uv}^{s,a}\notag\\
 &\le\sum_{u\in N(a)}w_{ua}\sum_{v\in V}e_{uv}^{s,a}\notag\\
 &\le s\sqrt{\frac{\ln2}{2(r-2)}}
             \sum_{u\in N(a)}w_{ua}.
             \label{eq:seed-weighted-error}
\end{align}
The last step uses the row-wise bound
\eqref{eq:conditional-pair-error}. The factor $1/2$ converts
ordered pairs to unordered pairs; diagonal terms contribute
zero because $w_{uu}=0$.

The first transition gives
\begin{equation}\label{eq:first-seed-probability}
 \Prb[s,u_1=a]
   =\frac{y_\emptyset^s}{y_\emptyset}
       \frac{y_a^s}{s y_\emptyset^s}
   =\frac{y_a^s}{s y_\emptyset}.
\end{equation}
For every pair $uv$, \eqref{eq:bdd-size} with $S=\{u,v\}$
gives the exact averaging identity
\begin{align}
 &\sum_{s,a:\,y_a^s>0}
       \Prb[s,u_1=a]\frac{y_{auv}^s}{y_a^s}\notag\\
 &\qquad=\frac1{y_\emptyset}
       \sum_s\frac1s\sum_{a\in V}y_{auv}^s
  =\frac{y_{uv}}{y_\emptyset}.
  \label{eq:first-seed-pair-identity}
\end{align}
Therefore convexity of absolute value, followed by
\eqref{eq:seed-weighted-error}, shows
\begin{align*}
 &\sum_{uv\in\binom{V}{2}}w_{uv}
   \left|\Prb[u,v\in C]-\frac{y_{uv}}{y_\emptyset}\right|\\
 &\quad\le\sum_{s,a}\Prb[s,u_1=a]
             \sum_{uv\in\binom{V}{2}}w_{uv}e_{uv}^{s,a}\\
 &\quad\le\frac1{y_\emptyset}
       \sqrt{\frac{\ln2}{2(r-2)}}
       \sum_{a\in V}\sum_s y_a^s
                  \sum_{u\in N(a)}w_{au}\\
 &\quad=\frac{2w(E_{\mathrm{adm}})}{y_\emptyset}
        \sqrt{\frac{\ln2}{2(r-2)}}\\
 &\quad\le\frac{\eps\OPT}{y_\emptyset}.
\end{align*}
The equality uses $\sum_s y_a^s=1$ and counts each admissible edge
twice. The choice of $r$ in \eqref{eq:r-def} gives
\begin{equation}\label{eq:bdd-r-accounting}
  2R_\eps\sqrt{\frac{\ln2}{2(r-2)}}\le\eps.
\end{equation}
Together with $w(E_{\mathrm{adm}})\le R_\eps\OPT$, this proves the
last inequality.
\end{proof}

\paragraph{Expected cost and range.}

The cost functional is defined in \eqref{eq:cluster-cost}.

\begin{lemma}[Expected one-cluster cost]\label{lem:expected-cluster-cost}
The sampled cluster satisfies
\[
 \E[\cost(C)]\le\frac{(1+2\eps)\OPT}{y_\emptyset}.
\]
\end{lemma}
\begin{proof}
Write $p_{uv}:=\Prb[u,v\in C]$. For a positive pair $uv$,
its contribution to $\E[\cost(C)]$ is
\[
 \frac{w_{uv}}2\bigl(\Prb[u\in C]+\Prb[v\in C]-2p_{uv}\bigr)
 =w_{uv}\left(\frac1{y_\emptyset}-p_{uv}\right),
\]
using \Cref{lem:one-marginal}. For a negative pair $uv$,
its contribution is $w_{uv}p_{uv}$. Thus
\begin{align*}
 \E[\cost(C)]
 &=\sum_{uv\in E^+}w_{uv}\left(\frac1{y_\emptyset}-p_{uv}\right)
   +\sum_{uv\in E^-}w_{uv}p_{uv}\\
 &\le\frac1{y_\emptyset}
    \left(\sum_{uv\in E^+}w_{uv}(1-y_{uv})
          +\sum_{uv\in E^-}w_{uv}y_{uv}\right)\\
 &\qquad+\sum_{uv\in\binom{V}{2}}w_{uv}
               \left|p_{uv}-\frac{y_{uv}}{y_\emptyset}\right|\\
 &\le\frac{(1+2\eps)\OPT}{y_\emptyset},
\end{align*}
where the last line combines
\eqref{eq:sampling-input-bounds} and \Cref{lem:weighted-error}.
\end{proof}

Concentration also requires a deterministic range bound.

\begin{lemma}[One-cluster range]\label{lem:cluster-range}
Every output of \Cref{alg:sample} satisfies
\[
 0\le\cost(C)\le\obj(\Kcal)+n\,w(E_{\mathrm{adm}}).
\]
Consequently,
$\cost(C)\le(\alpha+nR_\eps)\OPT$.
\end{lemma}
\begin{proof}
Nonnegativity gives the lower bound. Positive pairs cut by $C$ lie
between atoms, while negative pairs inside individual atoms are paid
by the atom partition. Their combined contribution is at most
$\obj(\Kcal)$.

Let $a$ be the first seed and $U=C\setminus K_a$.
By \Cref{lem:sampling-seed-support}, every pair between $U$ and $K_a$
is admissible. Remaining negative inter-atom pairs contribute at most
\[
  w(U,K_a)+\frac12w(U,U).
\]
If $U\ne\emptyset$, the set-wise triangle inequality gives
\[
  \frac12w(U,U)
    \le\frac{|U|}{|K_a|}\,w(U,K_a).
\]
The same bound trivially holds for $U=\emptyset$. Therefore
\begin{align*}
  w(U,K_a)+\frac12w(U,U)
  &\le\frac{|C|}{|K_a|}\,w(U,K_a)\\
  &\le n\,w(E_{\mathrm{adm}}),
\end{align*}
where $|C|\le n$, $|K_a|\ge1$, and every pair joining
$U$ to $K_a$ is admissible. In particular,
$w(U,K_a)\le w(E_{\mathrm{adm}})$.
Add the contributions and apply \Cref{thm:preclustering}.
\end{proof}

\section{Constructing the cluster-LP solution}\label{sec:construction}

We complete the construction by taking independent samples and
using their frequencies as cluster-LP masses. The sampler satisfies
\begin{equation}\label{eq:section6-sampler-recap}
 \Prb[u\in C]=\frac1{y_\emptyset}\quad(u\in V),
 \qquad
 \E[\cost(C)]\le\frac{(1+2\eps)\OPT}{y_\emptyset}.
\end{equation}
We scale down empirical masses and fill deficits with atoms.
Chernoff and Hoeffding bounds will control coverage and cost on the
original independent samples; we intersect their events by a union bound.

\paragraph{Conditioning on the preclustering step.}
Fix an outcome of the initialization $\Ccal^0$ such that $\obj(\Ccal^0)\leq\rho \OPT$.  With this outcome fixed,
\Cref{thm:preclustering,eq:R-def} supply an atom partition $\Kcal$
and an admissible edge set $E_{\mathrm{adm}}$ with
\begin{equation}\label{eq:section6-structural-bounds}
   \obj(\Kcal)\le\alpha\OPT,
   \qquad w(E_{\mathrm{adm}})\le R_\eps\OPT.
\end{equation}
Fix an optimal bounded sub-cluster LP solution $y$, with $1\le y_\emptyset\le n$.
The LP is feasible even after unsuccessful initialization, since the
atom partition induces a feasible solution; only its near-optimality
uses $\mathcal E_{\mathrm{init}}$. All probability bounds below are
uniform over fixed successful upstream outcomes and concern fresh
independent calls to \Cref{alg:sample}.

Set
\[
  \eps_1:=\frac{\eps}{\alpha}\in(0,1],
  \qquad
  \lambda:=\frac{y_\emptyset}{(1+\eps_1)M},
\]
where $M$ is chosen in \eqref{eq:M-choice}. Store masses in a sparse
dictionary, with all unspecified entries zero.

\begin{algorithm}[htbp]
\caption{ConstructClusterLPSolution}\label{alg:construct-z}
\begin{algorithmic}[1]
\Require An optimal bounded sub-cluster LP solution $y$, the atom partition $\Kcal$,
         an integer $M$, and $\eps_1=\eps/\alpha$
\Ensure A feasible finite-support cluster-LP vector $z$, or an abort signal
\State Initialize an empty dictionary $z$ (all unspecified entries are zero)
\For{$i=1,\ldots,M$}
  \State Independently draw $C_i\gets\Call{SampleOneCluster}{}$
  \State $z_{C_i}\gets z_{C_i}+y_\emptyset/((1+\eps_1)M)$
\EndFor
\ForAll{$K\in\Kcal$}
  \State $\Delta_K\gets 1-\sum_{S\supseteq K}z_S$
  \If{$\Delta_K<0$ or $\Delta_K>2\eps_1/(1+\eps_1)$}
    \State \textbf{abort}
  \EndIf
  \State $z_K\gets z_K+\Delta_K$
\EndFor
\State \Return $z$
\end{algorithmic}
\end{algorithm}

All support sets are unions of atoms. The target cluster LP does not
require their atoms to be mutually admissible.

\paragraph{Feasibility.}

For $K\in\Kcal$ define the occupancy count
\[
   N_K:=\sum_{i=1}^M\one_{\{K\subseteq C_i\}}.
\]
Each vertex of $K$ has empirical coverage $\lambda N_K$, so
\begin{equation}\label{eq:coverage-deficit}
   \Delta_K
      =1-\frac{y_\emptyset N_K}{(1+\eps_1)M}.
\end{equation}
Corrections to other atoms cannot affect this value because distinct
atoms are disjoint.

\begin{lemma}[Exact feasibility after atom corrections]
\label{lem:section6-feasible}
If \Cref{alg:construct-z} does not abort, its output $z$ is a
nonnegative feasible cluster-LP solution.  It has at most
$M+|\Kcal|$ nonzero variables, and its supported sets are unions of
atoms.  In particular, every pair of vertices in the same atom has
fractional separation value zero.
\end{lemma}
\begin{proof}
No abort means all corrections are nonnegative. Since support sets
preserve atoms and distinct atoms are disjoint, for $u\in K$,
\[
  \sum_{S\ni u}z_S
   =\frac{y_\emptyset N_K}{(1+\eps_1)M}+\Delta_K
   =1.
\]
For each unordered pair $uv$, define
\[
   x_{uv}:=1-\sum_{S\supseteq\{u,v\}}z_S.
\]
Nonnegativity and vertex coverage give
\[
  0\le\sum_{S\supseteq\{u,v\}}z_S
       \le\sum_{S\ni u}z_S=1.
\]
Hence \eqref{eq:clp-vertex}, \eqref{eq:clp-pair}, and \eqref{eq:clp-nonneg} hold.
Vertices in one atom belong to exactly the same support sets, giving
$x_{uv}=0$. Combining duplicate sampled sets leaves at most $M$ such
sets and $|\Kcal|$ correction sets.
\end{proof}

By \Cref{lem:one-marginal}, for $u\in K$,
\[
    \Prb[K\subseteq C_i]=\Prb[u\in C_i]=\frac1{y_\emptyset}.
\]
Since the samples are independent,
\[
   N_K\sim\operatorname{Bin}\!\left(M,\frac1{y_\emptyset}\right),
   \qquad \mu:=\E[N_K]=\frac{M}{y_\emptyset}.
\]
Counts for different atoms may be dependent; each individual count
is a sum of independent indicators.

For reference, the standard multiplicative Chernoff estimate for a
sum $X$ of independent Bernoulli variables with mean $\mu$ is
\begin{equation}\label{eq:chernoff}
  \Prb[|X-\mu|>\delta\mu]
    \le2\exp(-\delta^2\mu/3)
  \qquad(0\le\delta\le1).
\end{equation}
From \eqref{eq:coverage-deficit}, the exact non-abort condition for
atom $K$ is
\begin{align}
  0\le\Delta_K\le\frac{2\eps_1}{1+\eps_1}
  &\iff
  1-\eps_1\le\frac{y_\emptyset N_K}{M}\le1+\eps_1\notag\\
  &\iff
  (1-\eps_1)\mu\le N_K\le(1+\eps_1)\mu.
  \label{eq:section6-count-range}
\end{align}
Apply \eqref{eq:chernoff} with $\delta=\eps_1$ and union-bound over atoms:
\begin{equation}\label{eq:abort-prob}
  \Prb[\text{\Cref{alg:construct-z} aborts}]
  \le2|\Kcal|\exp\!\left(-\frac{\eps_1^2M}{3y_\emptyset}\right)
  \le2n\exp\!\left(-\frac{\eps_1^2M}{3y_\emptyset}\right).
\end{equation}

\paragraph{Objective bound.}

On every non-aborting outcome, \eqref{eq:cluster-cost-decomp} gives
\begin{equation}\label{eq:section6-exact-cost}
  \obj(z)
    =\frac{y_\emptyset}{1+\eps_1}\cdot
       \frac1M\sum_{i=1}^M\cost(C_i)
      +\sum_{K\in\Kcal}\Delta_K\cost(K).
\end{equation}
This also accounts for repeated sets and sampled atoms, since masses
add linearly. As $\Kcal$ is a partition,
\begin{equation}\label{eq:section6-atom-cost-sum}
  \sum_{K\in\Kcal}\cost(K)=\obj(\Kcal)\le\alpha\OPT.
\end{equation}
Using $0\le\Delta_K\le2\eps_1/(1+\eps_1)$ and nonnegative cluster costs,
\begin{align}
  \obj(z)
  &\le\frac{y_\emptyset}{1+\eps_1}
            \frac1M\sum_{i=1}^M\cost(C_i)
       +\frac{2\eps_1}{1+\eps_1}\obj(\Kcal)\notag\\
  &\le\frac{y_\emptyset}{1+\eps_1}
            \frac1M\sum_{i=1}^M\cost(C_i)
       +\frac{2\eps}{1+\eps_1}\OPT.
  \label{eq:z-obj-preconc}
\end{align}

Let $Z_i:=\cost(C_i)$.  By
\Cref{lem:expected-cluster-cost,lem:cluster-range} and the fixed
structural bounds in \eqref{eq:section6-structural-bounds}, the $Z_i$
are independent, identically distributed variables satisfying
\begin{equation}\label{eq:section6-cost-bounds}
  0\le Z_i\le B_\eps\OPT,
  \qquad
  \E[Z_i]\le\frac{(1+2\eps)\OPT}{y_\emptyset},
  \qquad B_\eps:=\alpha+nR_\eps.
\end{equation}
 For independent variables $X_i\in[a_i,b_i]$, Hoeffding's inequality
states
\begin{equation}\label{eq:hoeffding}
  \Prb\left[\sum_iX_i-\E\sum_iX_i>t\right]
  \le\exp\!\left(-\frac{2t^2}{\sum_i(b_i-a_i)^2}\right).
\end{equation}
Suppose first that $\OPT>0$.  Applying
\eqref{eq:hoeffding} to the $Z_i$ with
$t=M\eps\OPT/y_\emptyset$, and using the expectation bound in
\eqref{eq:section6-cost-bounds}, yields
\begin{align}
 &\Prb\left[
   \frac1M\sum_{i=1}^M\cost(C_i)
      >\frac{(1+3\eps)\OPT}{y_\emptyset}\right]\notag\\
 &\qquad\le
   \exp\!\left(-\frac{2\eps^2M}
        {y_\emptyset^2(\alpha+nR_\eps)^2}\right).
   \label{eq:cost-concentration}
\end{align}

Define the \emph{coverage event} $\mathcal G_{\mathrm{cov}}$ to mean
that the algorithm does not abort, and the \emph{cost event}
$\mathcal G_{\mathrm{cost}}$ to mean that the empirical average is at
most $(1+3\eps)\OPT/y_\emptyset$.  On their intersection,
\eqref{eq:z-obj-preconc} gives
\begin{align*}
  \obj(z)
  &\le\frac{(1+3\eps)\OPT}{1+\eps_1}
       +\frac{2\eps\OPT}{1+\eps_1}\\
  &=\frac{1+5\eps}{1+\eps_1}\OPT
   \le(1+5\eps)\OPT.
\end{align*}
Both events concern the original samples; the cost bound is not
conditioned on $\mathcal G_{\mathrm{cov}}$.

\paragraph{Sample size and success probability.}

Since $1\le y_\emptyset\le n$ and $R_\eps\ge\max\{1,\alpha\}$,
\[
   \alpha+nR_\eps\le 2nR_\eps.
\]
Choose
\begin{equation}\label{eq:M-choice}
  M:=\left\lceil
       C_0\,\frac{n^4R_\eps^2\ln(8n)}{\eps^2}
      \right\rceil,
  \qquad C_0:=6.
\end{equation}
This uses no unknown ratio $w(E_{\mathrm{adm}})/\OPT$ and gives
$M=O(n^4\log n\,\eps^{-984})$. Since $\eps_1=\eps/\alpha$ and
$R_\eps\ge\alpha$,
\begin{align*}
  \frac{\eps_1^2M}{3y_\emptyset}
  &\ge \frac{\eps^2}{3\alpha^2n}
       \cdot\frac{6n^4R_\eps^2\ln(8n)}{\eps^2}
   \ge 2n^3\ln(8n),\\
  \frac{2\eps^2M}
       {y_\emptyset^2(\alpha+nR_\eps)^2}
  &\ge\frac{2\eps^2}{n^2(2nR_\eps)^2}
       \cdot\frac{6n^4R_\eps^2\ln(8n)}{\eps^2}
   =3\ln(8n).
\end{align*}
Consequently, \eqref{eq:abort-prob} and \eqref{eq:cost-concentration} yield
\begin{align}
 \Prb[\mathcal G_{\mathrm{cov}}^c]
   &\le 2n(8n)^{-2n^3}\le\frac1{32n},\notag\\
 \Prb[\mathcal G_{\mathrm{cost}}^c]
   &\le(8n)^{-3}.
 \label{eq:section6-final-conditional-prob}
\end{align}
The bounds hold for every fixed successful initialization. Averaging
over upstream randomness and adding its failure probability gives
\begin{equation}\label{eq:section6-overall-failure}
 \Prb[\text{construction fails}]
 \le\frac1{8n}+\frac1{32n}+\frac1{(8n)^3}
 <\frac1{4n}.
\end{equation}
 \begin{theorem}[Finite-support cluster-LP construction]
\label{thm:clusterlp-5eps}
For every fixed $\eps\in(0,1]$, the algorithm consisting of
initialization, spectral preclustering, the bounded sub-cluster LP,
and \Cref{alg:construct-z} runs in randomized polynomial time and
returns, with probability at least $1-1/(2n)$, an explicitly
represented feasible cluster-LP solution $z$ such that
\[
   \obj(z)\le(1+5\eps)\OPT,
   \qquad |\operatorname{supp}(z)|\le M+|\Kcal|.
\]
The guarantees hold for arbitrary nonnegative pseudometric weights,
with no bounded-weight-ratio assumption.
\end{theorem}
\begin{proof}
On successful initialization and
$\mathcal G_{\mathrm{cov}}\cap\mathcal G_{\mathrm{cost}}$,
\Cref{lem:section6-feasible,eq:z-obj-preconc} give feasibility, support,
and cost. \eqref{eq:section6-overall-failure} bounds total failure by
$1/(4n)<1/(2n)$. For fixed $\eps$, $r=O(\eps^{-984})$ is constant and
the bounded sub-cluster LP has size $n^{O(r)}$. The sample count is polynomial in
$n$, as are each call to the sampler, the dictionary updates, and the
corrections. Rational LP solutions have polynomial encoding length,
giving the stated $n^{\operatorname{poly}(1/\eps)}$ running time.
\end{proof}

\paragraph{Parameter implementation.}
The logarithms in the displayed parameter choices need not be
represented exactly. One may increase $r$ to
$2+\lceil2R_\eps^2/\eps^2\rceil$ and replace $\ln(8n)$ in $M$
by $\lceil\log_2(8n)\rceil$. Both changes only increase the sample
budgets, preserve all estimates, and leave the stated polynomial
dependence unchanged. The probabilities used by the sampler and the
coverage corrections are rational when the input and the LP solution
are rational.

\paragraph{From construction to the two main theorems.}
For \Cref{thm:clusterlp-main}, use internal accuracy $\eps=\eta/5$.
Then $\obj(z)\le(1+\eta)\OPT$ with failure at most $1/(2n)$.

For \Cref{thm:main}, let $\eta\in(0,1]$ be the desired approximation
slack, and invoke the preceding construction with cluster-LP accuracy
$\eta/8$.  This gives $\obj(z)\le(1+\eta/8)\OPT$ except with
probability at most $1/(2n)$.  Conditional on such a feasible output,
\Cref{lem:rounding} yields an independent rounding with expected
clustering cost at most $2\obj(z)$.  By Markov's inequality, for
$\obj(z)>0$ a single rounding has probability at least
\[
    \frac{\eta/2}{2+\eta/2}
\]
of producing a clustering with cost at most
$(2+\eta/2)\obj(z)$.  Repeating the rounding independently
$\left\lceil(2+\eta/2)(\eta/2)^{-1}\ln(2n)\right\rceil$
times, and keeping the least expensive outcome, reduces this
conditional failure probability to at most $1/(2n)$.  For $\obj(z)=0$,
nonnegativity and the expectation bound imply that every rounding
has cost zero almost surely.  Finally, for $0<\eta\le1$,
\[
  (2+\eta/2)(1+\eta/8)
  =2+\frac{3\eta}{4}+\frac{\eta^2}{16}
  \le2+\eta.
\]
A union bound over construction and rounding failures proves the
$(2+\eta)$-approximation with probability at least $1-1/n$.
For a requested accuracy larger than one, it suffices to use
accuracy one throughout.  Instances with at most one vertex are
handled by their unique partition.

\section*{Discussion}

Two weighted estimates make the cluster-LP construction possible.
The zero-diagonal atom kernel inherits factor-two Ptolemy, giving a
degree--volume bound, a warm start, and a computable admissibility
certificate of total weight $O_\eps(\OPT)$. The original vertex-level
triangle inequality then charges row-wise sampling error through the
first seed to that admissible weight. The two arguments use different
weight spaces: the atom kernel need not be a pseudometric.

The resulting feasible solution is near-optimal relative to the
\emph{integral} clustering cost; it is not an approximate separation
procedure for the fractional optimum. The main quantitative limitation
is $R_\eps=O(\eps^{-491})$ and hence $r=O(\eps^{-984})$. A more efficient
local-to-global comparison or a construction avoiding the high-order
relaxation could improve this dependence. Finally, factor-$2$ rounding
uses only nonnegative weights. Metric-sensitive rounding, or a sharper
analysis of the solutions produced here, may improve the approximation;
no matching hardness threshold at two is established by this work.

\paragraph{AI methodology.}
We used ChatGPT Pro 5.5 to assist with a part of the proof and to polish the paper.
The tool materially affected \Cref{subsec:nearopt}, which is reviewed and verified by the authors.
We used GPT-6.1 Sol and GPT-6 Astra to do a final review and to polish the paper.

\clearpage
\appendix
\section{Exchange identities and boundary cases}\label{app:atoms}

For reference, we collect the exact cost changes used in
\Cref{lem:atom-respecting}. Its fragment and strict-extraction arguments
remain in the main proof.

\begin{lemma}[Costs of the three exchange operations]
\label{lem:app-exchange-identities}
Let $\Ccal$ be a clustering and let $K\subseteq V$ be an atom.
\begin{enumerate}[label=(\roman*)]
\item If $A=C\cap K$ and $Q=C\setminus K$ are both nonempty,
  splitting the cluster $C$ into $A$ and $Q$ changes the objective by
  $w^+(A,Q)-w^-(A,Q)$.
\item If $u$ moves from its current cluster $C$ to a different
  cluster $C'$, the change is
  $\operatorname{score}_u(C)-\operatorname{score}_u(C')$, where
  $\operatorname{score}_u(T)=w^+(u,T)-w^-(u,T)$. For a new singleton,
  take $C'=\emptyset$ and $\operatorname{score}_u(C')=0$.
\item Let $C_1,\ldots,C_\ell$ be the clusters meeting $K$ and set
  $A_i=C_i\cap K$, $Q_i=C_i\setminus K$, and $B_i=K\setminus C_i$.
  Removing $K$ from these clusters and making $K$ one cluster
  changes the objective by
  \begin{align}
  \Delta_K
  &=\sum_{i=1}^\ell\bigl(w^+(A_i,Q_i)-w^-(A_i,Q_i)\bigr)
       \notag\\
  &\quad+\frac12\sum_{i=1}^\ell
       \bigl(w^-(A_i,B_i)-w^+(A_i,B_i)\bigr).
  \label{eq:app-exchange-exact}
  \end{align}
\end{enumerate}
\end{lemma}
\begin{proof}
For (i), only pairs across $A,Q$ change status. For (ii), pairs from
$u$ to its old cluster become separated and pairs to the destination
become co-clustered; $w_{uu}=0$ permits including $u$ in its old score.
For (iii), pairs in $A_i\times Q_i$ become separated, and pairs in
different fragments become co-clustered. The latter occur twice in
$\sum_i w(A_i,B_i)$, explaining the factor $1/2$; no other pair changes.
\end{proof}

\paragraph{Boundary cases.}
For a nonsingleton atom with $w(K,K)=0$, \Cref{lem:atom-stability}
gives $w^+(u,V\setminus K)+w^-(u,K)=0$ for every $u\in K$.
Extraction therefore cannot increase cost, and disjointness preserves
other intact atoms. For $w(K,K)>0$, \Cref{prop:settriangle} makes
$w(u,K)>0$ and $w(A,B)>0$ for every nontrivial split $K=A\uplus B$.
The fragment proof handles $w(A,A)=0$ separately before dividing by it.
Finally, \eqref{eq:s3-extraction-budget} is strict even at the permitted
upper endpoint of $\beta$, so \eqref{eq:s3-extraction-change} remains
negative. These cover all zero-distance and equality cases in the
exchange proof.

\section{Local distributions and consequences of the bounded LP}
\label{app:bdd-local-laws}

We verify the local probability laws, size support, and preservation
of atom and admissibility constraints used in
\Cref{sec:boundedlp,sec:sampling}. All moments below belong to a
feasible solution of the bounded sub-cluster LP.

\begin{lemma}[Local probability laws]\label{lem:bdd-local-laws}
Fix $s\in[n]$. For any $W\subseteq V$ with $|W|\le r$ and
$A\subseteq W$, define the unnormalized mass
\begin{equation}\label{eq:bdd-local-mass}
  M_W^s(A)
  :=\sum_{B\subseteq W\setminus A}(-1)^{|B|}y_{A\cup B}^s.
\end{equation}
Then $M_W^s(A)\ge0$, $\sum_{A\subseteq W}M_W^s(A)=y_\emptyset^s$,
and these masses are consistent under restriction to subsets of
$W$. Consequently, whenever $y_\emptyset^s>0$, the normalized
masses $M_W^s(A)/y_\emptyset^s$ define a consistent family of
probability distributions on $\{0,1\}^W$. In these distributions,
the unnormalized mass that every vertex in $S\subseteq W$ is
present equals $y_S^s$.
\end{lemma}
\begin{proof}
Taking $S=A$ and $T=W\setminus A$ in
\eqref{eq:bdd-SA} shows that $M_W^s(A)\ge0$. For normalization,
interchange the finite sums:
\begin{align*}
  \sum_{A\subseteq W}M_W^s(A)
  &=\sum_{U\subseteq W}y_U^s
       \sum_{A\subseteq U}(-1)^{|U\setminus A|}
   =y_\emptyset^s.
\end{align*}
Indeed, the inner sum is zero unless $U=\emptyset$.
For $v\in W$ and $A\subseteq W\setminus\{v\}$, separating the
terms that include $v$ from those that do not gives
\[
  M_W^s(A)+M_W^s(A\cup\{v\})
  =M_{W\setminus\{v\}}^s(A).
\]
Iterating this identity proves consistency under every
restriction. In particular, summing the masses of all exact
assignments on $W$ that contain $S$ gives $y_S^s$.
Dividing by $y_\emptyset^s$ when it is positive gives the stated
local probability laws. If $y_\emptyset^s=0$, nonnegativity and
consistency imply that every $y_S^s$ is zero, so this size never
needs to be sampled.
\end{proof}

The local laws imply the monotonicity in
\eqref{eq:bdd-monotonicity}: requiring more vertices to be included
cannot increase their joint inclusion probability. The unnormalized
mass $y_\emptyset^s$ may exceed one.

\begin{lemma}[Support of size-indexed moments]\label{lem:bdd-size-support}
For every $s\in[n]$ and $S\subseteq V$ with $|S|\le r$ and
$|S|>s$, one has $y_S^s=0$.
\end{lemma}
\begin{proof}
If $1\le |S|\le r-1$, the fixed-size identity gives
\[
  (s-|S|)y_S^s=\sum_{u\notin S}y_{Su}^s\ge0.
\]
Since $|S|>s$ and $y_S^s\ge0$, this forces $y_S^s=0$.
For $|S|=r>s$, take $S'\subset S$ with $|S'|=r-1$.
If $s<r-1$, the preceding case gives $y_{S'}^s=0$ and
monotonicity gives $y_S^s=0$. If $s=r-1$, the fixed-size
identity at $S'$ gives
$\sum_{u\notin S'}y_{S'u}^s=0$, which again forces
$y_S^s=0$. These are the only possibilities.
\end{proof}

\begin{lemma}[Propagation of atom and admissibility constraints]
\label{lem:bdd-atom-propagation}
For every size $s$, the following statements hold.

\emph{(i) Atom consistency.} If $u$ and $v$ belong to the same
atom, then $y_u^s=y_v^s=y_{uv}^s$. Moreover, for every
$U\subseteq V$ with $|U|\le r-2$,
\begin{equation}\label{eq:bdd-conditional-atom}
  y_{Uu}^s=y_{Uv}^s=y_{Uuv}^s.
\end{equation}
In particular, if $U$ intersects that atom and $|U|\le r-1$,
then $y_{Uv}^s=y_U^s$ for every vertex $v$ of that atom.

\emph{(ii) Exclusion of forbidden pairs.} If $a,b$ belong to
distinct atoms and $ab\notin E_{\mathrm{adm}}$, then
$y_{ab}^s=0$. For every $U\supseteq\{a\}$ with
$|U\cup\{b\}|\le r$,
\begin{equation}\label{eq:bdd-conditional-exclusion}
  y_{Ub}^s=0.
\end{equation}
\end{lemma}
\begin{proof}
For (i), the nonnegative differences $y_u^s-y_{uv}^s$ and
$y_v^s-y_{uv}^s$ sum over $s$ to zero by
\eqref{eq:bdd-single} and \eqref{eq:bdd-atom}; hence each is zero.

In the local law on $U\cup\{u,v\}$, the events $u=1,v=0$ and
$u=0,v=1$ have zero mass. Requiring all vertices of $U$ to be present
cannot increase that mass, proving \eqref{eq:bdd-conditional-atom}.
If $a\in U$ is in the same atom as $v$, the zero mass of $a=1,v=0$
similarly gives $y_U^s=y_{Uv}^s$ whenever $|U\cup\{v\}|\le r$.

For (ii), \eqref{eq:bdd-adm} gives
$\sum_s y_{ab}^s=y_{ab}=0$.
Since each summand is nonnegative, every $y_{ab}^s$ is zero.
Whenever $a\in U$ and $|U\cup\{b\}|\le r$,
monotonicity gives $0\le y_{Ub}^s\le y_{ab}^s=0$.
No assertion about a globally realized distribution is used.
\end{proof}

\begingroup
\footnotesize
\bibliographystyle{alpha}
\bibliography{new-refs}
\endgroup

\end{document}